\documentclass[pdftex,pra,aps,onecolumn,twoside,nofootinbib]{revtex4}

\usepackage[normalem]{ulem} 
\usepackage{braket}
\usepackage{graphicx}
\usepackage{tikz}
\usepackage{pgfplots}
\pgfplotsset{compat=newest}
\usepgfplotslibrary{fillbetween}
\usepackage{dsfont}
\usepackage[T1]{fontenc}
\usepackage[utf8]{inputenc}
\usepackage{float}
\usepackage{indentfirst}
\usepackage{enumerate}
\usepackage{url}
\usepackage{mathtools}
\usepackage{enumerate}
\usepackage{textcomp}
\usepackage{newcent}
\usepackage[sc,noBBpl]{mathpazo}
\usepackage{amsmath,amsfonts,amssymb,amsthm}
\usepackage[mathscr]{eucal}
\usepackage{tcolorbox}
\usepackage{hyperref}
\theoremstyle{plain}
\newtheorem{theorem}{Theorem}
\newtheorem{lemma}[theorem]{Lemma}
\newtheorem{proposition}[theorem]{Proposition}
\newtheorem{corollary}[theorem]{Corollary}

\newtheorem{definition}[theorem]{Definition}

\newtheoremstyle{note}{\topsep}{\topsep}{\slshape}{}{\scshape}{}{ }{}
\theoremstyle{note}

\newtheorem{remark}[theorem]{Remark}
\newtheorem{example}[theorem]{Example}

\newcommand{\<}{\langle}
\renewcommand{\>}{\rangle}

\newcommand\be{\begin{equation}}
\newcommand\ee{\end{equation}}
\newcommand\bea{\begin{array}}
	\newcommand\eea{\end{array}}
\newcommand\ben{\begin{eqnarray}}
\newcommand\een{\end{eqnarray}}
\newcommand\ot{\otimes}
\definecolor{forest}{RGB}{11, 102,35}

\newcommand{\A}{\mathcal{A}_{n}^{t_{n}}(d)}

\newcommand{\Hom}{\operatorname{Hom}}
\newcommand{\Span}{\operatorname{span}}

\usepackage{makecell}

\newcommand\bei{\begin{itemize}}
	\newcommand\eei{\end{itemize}}
\newcommand\bee{\begin{enumerate}}
	\newcommand\eee{\end{enumerate}}

\begin{document} 	
\title{From port-based teleportation to Frobenius reciprocity theorem: partially reduced irreducible representations and their applications}
	\author{Marek Mozrzymas$^1$, Micha{\l} Horodecki$^2$\footnote{michal.horodecki@ug.edu.pl, corresponding author}, Micha{\l} Studzi\'nski$^3$}
	\affiliation{
 $^1$Institute for Theoretical Physics, University of Wrocław,
  		50-204~Wrocław, Poland\\
    $^2$International Centre for Theory of Quantum Technologies, University of Gda{\'n}sk, 80-952, Poland\\
	    $^3$Institute of Theoretical Physics and Astrophysics, Faculty of Mathematics, Physics and Informatics,\\University of Gda\'nsk, Wita Stwosza 57, 80-308 Gda\'nsk, Poland}

\begin{abstract}
In this paper, we present the connection of two concepts as induced representation and partially reduced irreducible representations (PRIR) appear in the context of port-based teleportation protocols. 
Namely, for a given finite group $G$ with arbitrary subgroup $H$, we consider a particular case of matrix irreducible representations, whose restriction to the subgroup $H$, as a matrix representation of $H$, is completely reduced to diagonal block form with an irreducible representation of $H$ in the blocks. The basic properties of such representations are given. Then as an application of this concept, we show that the spectrum of the port-based teleportation operator acting on $n$ systems is connected in a very simple way with the spectrum of the corresponding Jucys-Murphy operator for the symmetric group $S(n-1)\subset S(n)$. This shows on the technical level relation between teleporation and one of the basic objects from the point of view of the representation theory of the symmetric group. This shows a deep connection between the central object describing properties of deterministic PBT schemes and objects appearing naturally in the abstract representation theory of the symmetric group. In particular, we present a new expression for the eigenvalues of the Jucys-Murphy operators based on the irreducible characters of the symmetric group.  As an additional but not trivial result, we give also purely matrix proof of the Frobenius reciprocity theorem for characters with explicit construction of the unitary matrix that realizes the reduction of the natural basis of induced representation to the reduced one. 
\end{abstract}
\maketitle

\section{Introduction}
Port-based teleportation (PBT)~\cite{ishizaka_asymptotic_2008,ishizaka_quantum_2009} is a remarkable protocol possessing a counter-intuitive emanating in that the teleported state requires no unitary correction and is ready for use after the sender performs a measurement and sends classical communication. No-unitary correction property has attracted wide attention from the community resulting in intensive research in the field of the PBT resulting in the development of quantum information theory. First, the PBT scheme offers a model of a universal programmable quantum processor~\cite{ishizaka_asymptotic_2008}, it gives a connection with quantum cryptography and instantaneous non-local computation~\cite{beigi2011simplified}. PBT protocols were instrumental in establishing a link between interaction complexity and entanglement in non-local computation and holography~\cite{may2022complexity}, between quantum communication complexity advantage and a
violation of a Bell inequality~\cite{buhrman_quantum_2016}, deriving fundamental bounds for quantum channels discrimination by designing PBT stretching protocols~\cite{pirandola2019fundamental}, and many other interesting results~\cite{pereira2021characterising,quintino2021quantum,PhysRevLett.122.170502}.
Describing the efficiency of the PBT protocol was a very complex task and for a long time satisfactory description, in particular, a description in higher dimensions and asymptotic behavior was missing. Difficulties came mostly from the point of view of mathematical methods and much effort was put into optimizing the protocol and its variants by developing a proper mathematical tool-kit~\cite{wang_higher-dimensional_2016,strelchuk_generalized_2013,Studzinski2017, StuNJP, MozJPA,christandl2021asymptotic,Leditzky2022} and recently in~\cite{grinko2023gelfandtsetlin}. 

In this paper, we start from Section~\ref{Gen} with a general description of the basic properties of {\it PBT-operator}, an object crucial for describing the efficiency of the deterministic PBT. We shortly remind its connection with the representation theory of the symmetric group. In particular, we focus on the aspect of the occurrence of induced representation and connection of the PBT-operator with Jucys-Murphy elements and partial transposition. We focus here also on the case of the multi-port-based teleportation protocols~\cite{StuIEEE} where we have to consider many layers of the induction process. Next, in Section~\ref{Sec:PRIR} we on the above-mentioned connection more formal. We present rigorously the concept of Partially Reduced Irreducible Representations (PRIRs)  - a notion related to subgroup adapted basis \cite{Koch_2012} - and we discuss their properties and their role in induced representation. As an application of developed techniques, we show how the spectrum of famous Jucys-Murphy elements~\cite{Mu, Ju} acting on $n-1$ systems is connected with the spectrum of the PBT operator acting on $n$ systems.  In particular, in Corollary~\ref{Sec2Cor2} we present new expressions for the spectrum of the Jucys-Murphy of the symmetric group by use of the irreducible characters. In this regard in Proposition~\ref{L:JMd2}, we present fully analytical expressions for the mentioned eigenvalues, in the natural representation of the Jucys-Murphy elements, when the dimension of the underlying space is 2. This is possible, since then we have analytical expressions for corresponding irreducible characters which are labeled by Young frames with up to two rows, see Lemma~\ref{Lemma:Harris}. In the same section, we also contain, up to our best knowledge, a completely unknown previously orthogonality relation. It is contained in Proposition~\ref{Sec2Prop2}.  Finally, we present purely matrix proof of the celebrated Frobenius reciprocity theorem for characters which was not known earlier in the literature. As an additional result, we give the explicit form of the unitary matrix $U(\beta )$ which realizes the reduction of the natural basis of induced representation to the reduced one. These results are contained in Theorem~\ref{Sec2Thm18}, preceded by auxiliary considerations.

\section{From PBT operator to induced representation and PRIR}
\label{Gen}
In this section, we will show
(based on our previous papers)  how the two concepts: the induced representation and 
PRIR appears in the context of port-based teleportation. We shall then show how this leads to the proof of the Frobenius reciprocity theorem for a subgroup of the symmetric group.

\subsection{Algebra of the symmetric group and partial transposition}
\label{secBrauer}
For the symmetric group $S(n)$, we can define its {\it natural} representation  $V^n_d:S(n)\rightarrow (\mathbb{C}^{d})^{\otimes n}$ by the following action on the set of basis vectors $\{|i\>\}_{i=1}^{d}$, where $d$ stand for the local dimension:
\begin{equation}
\label{eq:natrep}
	\forall \sigma \in S(n)\qquad V^n_d(\sigma ).|i_{1}\>\otimes |i_{2}\>\otimes
	\cdots \otimes |i_{n}\>=|i_{\sigma ^{-1}(1)}\>\otimes |i_{\sigma
			^{-1}(2)}\>\otimes \cdots \otimes |i_{\sigma ^{-1}(n)}\>.
\end{equation}
The representation $V^n_d(S(n))$  is defined in a given basis of the space $\mathbb{C}^d$, so it is a matrix representation, and operators $V^n_d(\sigma)$ just permute basis vectors according to the given permutation $\sigma$.
Whenever the numbers $n,d$ are clear from the context, we will write just simply $V(\sigma)$ instead of $V^n_d(\sigma)$ to simplify the notation.

For later purposes, we can introduce a matrix algebra $\mathcal{A}_d(n)$ spanned by the operators $V(\sigma)$ as follows:
\begin{equation}
\label{eq:groupAlg}
\mathcal{A}_n(d):=\operatorname{span}_{\mathbb{C}}\{V(\sigma) \ | \ \sigma\in S(n)\}.
\end{equation}
The above matrix algebra is just a natural representation of an abstract algebra $\mathbb{C}[S(n)]$ of the symmetric group $S(n)$. Having definition of the group algebra $\mathbb{C}[S(n)]$ in equation~\eqref{eq:groupAlg}, we can naturally introduce the algebra of partially transposed operators with respect to the last subsystem $\A$ in the following way:
\begin{definition}
	\label{def_A}
	For $\mathcal{A}_{n}(d):= \Span_{\mathbb{C}}\{V(\sigma ):\sigma \in S(n)\}$ we define a new complex algebra
	\be
 \label{eq:def_A}
	\mathcal{A}_{n}^{T_{n}}(d):= \Span_{\mathbb{C}}\{V^{T_{n}}(\sigma ) \ | \ \sigma \in S(n)\},
	\ee
	where the symbol $T_{n}$ denotes the partial transposition with respect to the last subsystem
	in the space $\Hom((\mathbb{C}^{d})^{\otimes n})$.
\end{definition}
Definition~\ref{def_A} can be extended to a larger number of partial transpositions, see for example~\cite{StuIEEE,grinko2023gelfandtsetlin}. The considered algebra $\mathcal{A}_{n}^{T_{n}}(d)$ is in fact a matrix representation of a diagram algebra called walled Brauer algebra $\mathcal{B}^{\delta}_{p,q}$, where $p,q\leq 0$, $p+q=n$, and $\delta \in \mathbb{C}$, introduced and analyzed in an abstract way in~\cite{VGTuraev_1990,KOIKE198957,BENKART1994529,BEN96,bulgakova:tel-02554375}.  The abstract algebra $\mathcal{B}^{\delta}_{p,q}$ is composed of formal combinations of diagrams. Each diagram has two rows with $p+q$ nodes, associated with a vertical wall between the first $p$ and the last $q$ nodes. These nodes are connected up in pairs in such a way that:
\begin{enumerate}
    \item if both nodes are in the same row, they must lie on different sides of the wall,
    \item if both nodes are in different rows, they must lie on the same side of the wall.
\end{enumerate}
 We illustrate the above construction with the notion of composition of such diagrams in Figure~\ref{fig:WBA}. 
\begin{figure}[h]
\includegraphics[width=.65\linewidth]{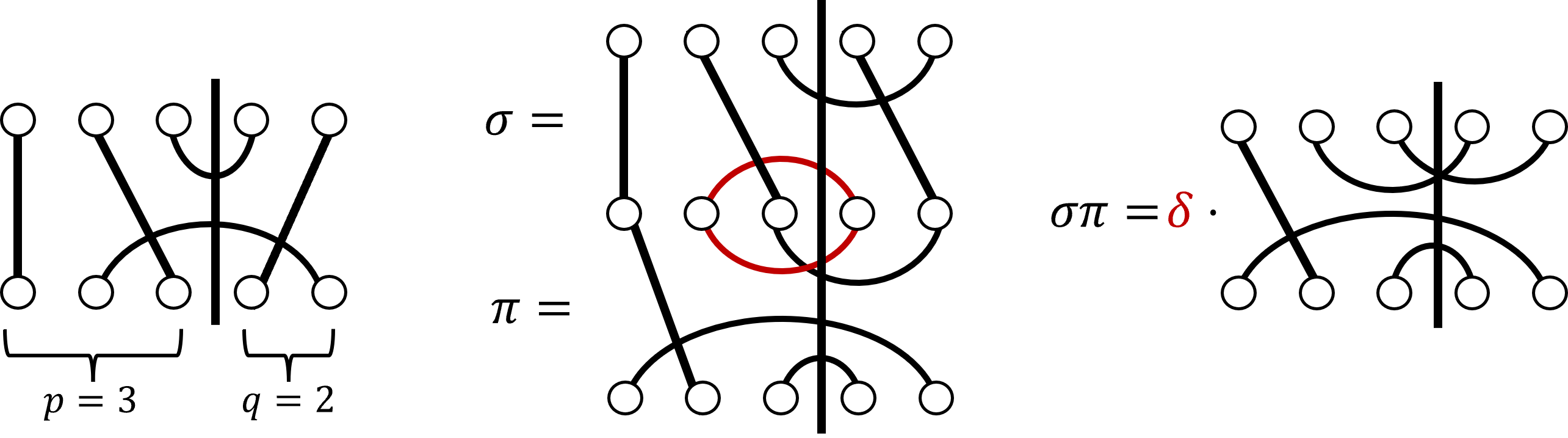}
	\caption{The left-hand side graphic presents an element from the walled Brauer algebra $\mathcal{B}^{\delta}_{3,2}$. On the right-hand side, we illustrate the composition of two diagrams  $\sigma,\pi \in \mathcal{B}^{\delta}_{3,2}$. When we identify a closed loop (red color), we multiply the resulting diagram by a scalar $\delta\in \mathbb{C}$. We see that the resulting composition $\sigma \pi$ also belongs to $\mathcal{B}^{\delta}_{3,2}$.}
	\label{fig:WBA}
\end{figure}
For any diagram from $\mathcal{B}^{\delta}_{p,q}$, the partial transposition $T$ can be understood by exchanging the nodes on the right-hand side of the wall. This procedure is illustrated in Figure~\ref{fig:WBA2}.
\begin{figure}[h]
\includegraphics[width=.35\linewidth]{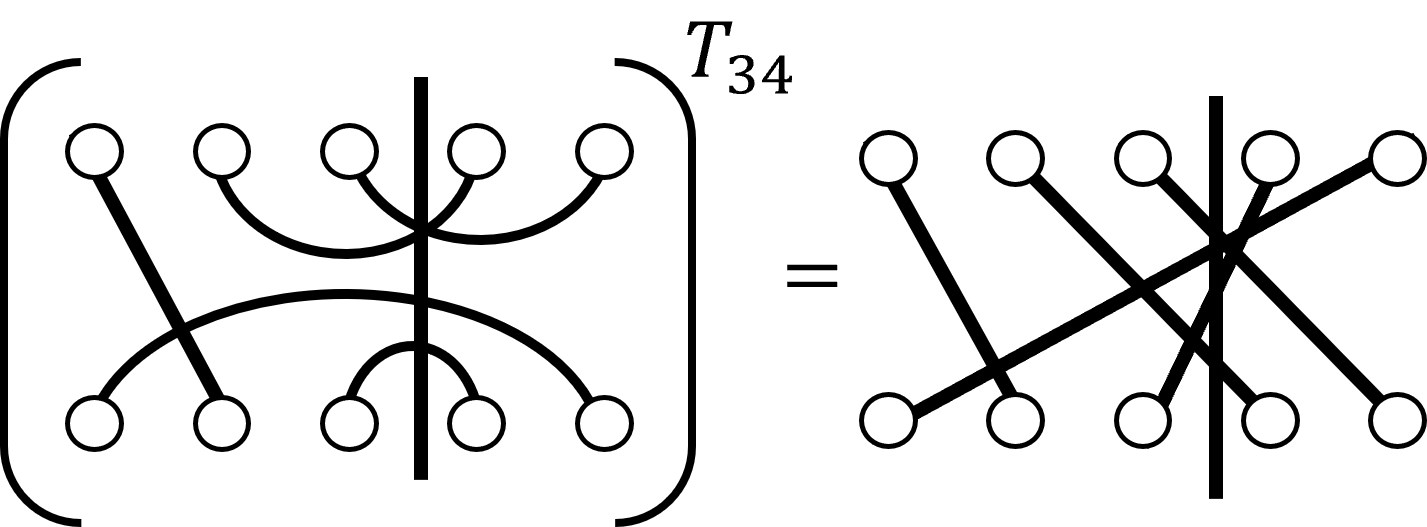}
	\caption{Graphic presents the procedure of partial transposition $T_{34}$ applied to  a diagram from $\mathcal{B}^{\delta}_{3,2}$. Notice that the resulting diagram is no longer an element of $\mathcal{B}^{\delta}_{3,2}$. The resulting diagram is an element of Brauer algebra $\mathcal{B}^{\delta}_{3+2}$ introduced in~\cite{Bra37} and graphically represents permutation from $S(3+2)=S(5)$. The walled Brauer algebra $\mathcal{B}^{\delta}_{p,q}$ is subalgebra of the Brauer algebra $\mathcal{B}^{\delta}_{p+q}$. We can say that $\pi \in \mathcal{B}^{\delta}_{p+q}$ if and only if $\pi^T \in S(p+q)$.}
	\label{fig:WBA2}
\end{figure}
To get matrix representation~\eqref{eq:def_A} we set $p=n-1,q=1$, and $\delta=d$. In this case of two systems the following relation between transposition $V((1,2))$ and unnormaalised projector on maximally entangled state $P_+=|\phi_+\>_{1,2}\<\phi_+|$ between systems 1 and 2:
\begin{equation}
\label{eq:maxent}
V^{T_2}((1,2))=dP_+,\qquad V^{T_2}((1,2))V^{T_2}((1,2))=dV^{T_2}((1,2))=dP_+,
\end{equation}
where $|\phi_+\>=\sum_{i=1}^d|i_1\>|i_2\>$. 
This relation will be exploited extensively in the further parts of this manuscript. This particular case is illustrated on the diagram level also in Figure~\ref{fig:WBA3}.
\begin{figure}[h]
\includegraphics[width=.53\linewidth]{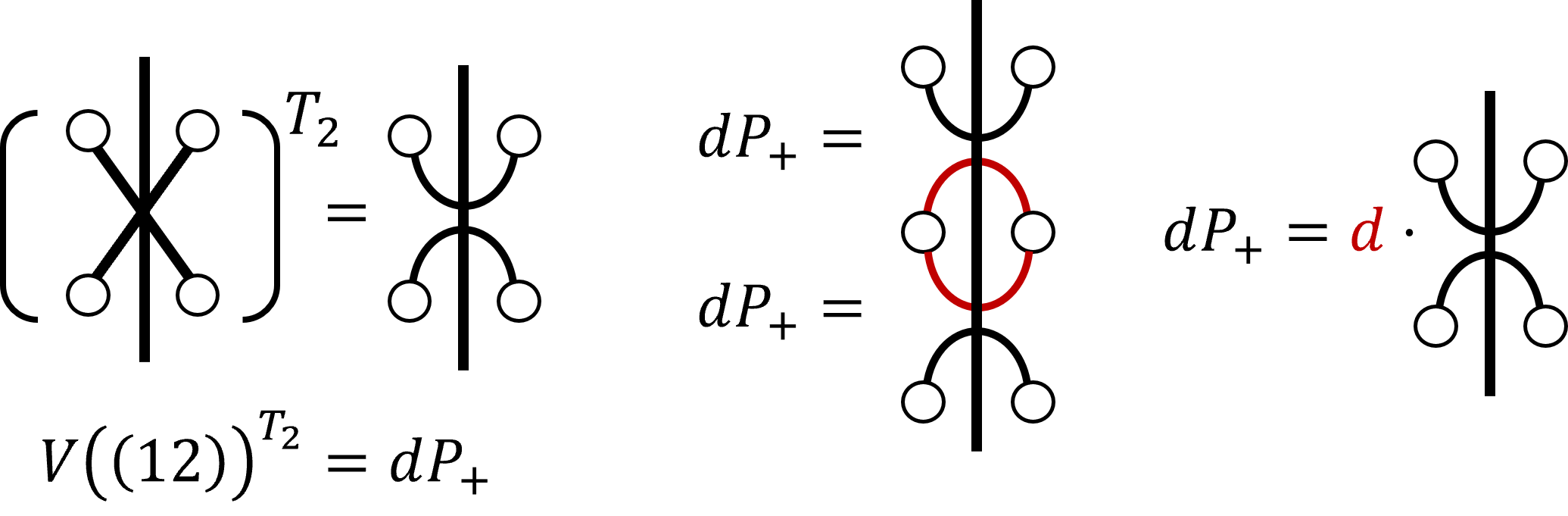}
	\caption{Graphic illustrates on the diagram-level relation given in equation~\eqref{eq:maxent}. This is a special case of the walled Brauer algebra $\mathcal{B}^{\delta}_{p,q}$ with $p=q=1$ and $\delta=d$, where $d$ is the local dimension.}
	\label{fig:WBA3}
\end{figure}

\subsection{PBT operator and induced representation}
One of the key ingredients of most of the results on port-based teleportation was understanding an operator that can be called  PBT operator. For standard port-based teleportation it acts on $n=N+1$ systems, where $N$ is the number of shared ports, and it is  of the form:
\begin{align}
\label{eq1}
    \rho=\sum_{a=1}^N
    V((a,N)) \left( \mathbf{1}^{\otimes N-1} \ot P_+\right)V((a,N)),
\end{align}
where $V(\sigma)$ is the operator that permutes systems according to permutation $\sigma$, $P_+=|\phi_+\>_{N,N+1}\<\phi_+|$ is unnormalized projector on the unnormalized maximally entangled state between systems $N$ and $N+1$,
\begin{align}
    |\phi_+\>=\sum_{i=1}^d |i\>_N|i\>_{N+1},
\end{align}
and $\mathbf{1}^{\otimes N-1}$ is the 
identity operator on rest of the systems.  

For generalization of the port-based teleportation for sending composite quantum systems~\cite{strelchuk_generalized_2013,stud2020A,Studzinski_2022,mozrzymas2021optimal} the relevant operator is a direct generalization of the above one:
\begin{align}
\label{eq2}    \rho=\frac{1}{(N-l)!}\sum_{\sigma\in S(N)}
V(\sigma)\left(\mathbf{1}^{\otimes N-l} \otimes P_+^{\otimes l}\right)V(\sigma^{-1}),
\end{align}
where $l$ is the number of systems to be teleported, and $S(N)$ is the symmetric group over $N$ elements. From Fig. \ref{fig:rho} it is easy to see on what systems identities and 
 are acting. The factor $1/(N-l)!$ is to remove overcounting.  
\begin{figure}[h]
\includegraphics[width=.8\linewidth]{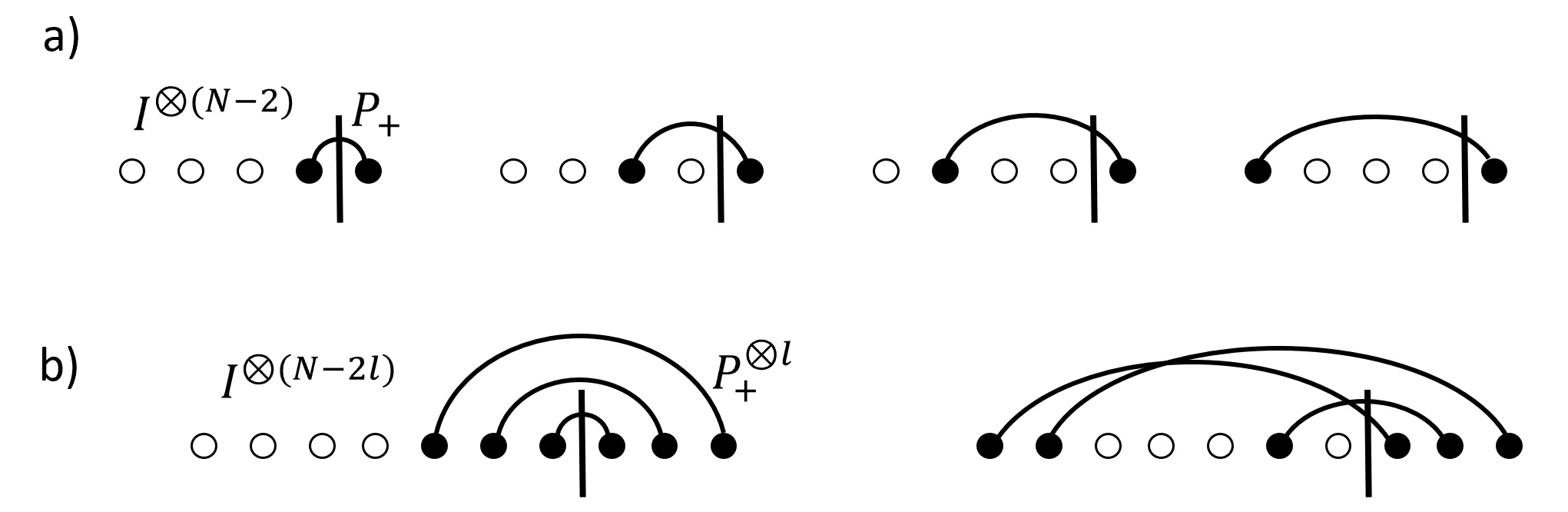}
	\caption{
 For $l=1$ the operator $\rho$ is a sum of the depicted operators in (a); here $N=4$. For $l>1$, $\rho$ is the sum of operators of type of those depicted in (b). Here $N=7$ and $l=3$.}
	\label{fig:rho}
\end{figure}
Thus the operator $\rho$ is sum of 
operators, for which $N-l$ systems are identity, while the other $2l$ are occupied by maximally entangled states. In \cite{MozJPA} it was noted that the operator is tightly related to the notion of {\it induced representation} as well as with 
PRIR ("partially reduced irreducible representation").

To see how they arise, let us first notice how permutation from $S(N)$
acts on the state. 
For each particular constituent of 
$\rho$ an effect will be the following: the ends of maximally entangled states that lie to the left of the "wall" will be 
redistributed, and permuted, and the same happens to the free systems. Let us start with one constituent of the state  - the one where all free systems are on the left:
\begin{align}
    I^{\ot N-l}\ot P_+^{\ot l}
\end{align}
Let us split the identity into projectors onto irreps of $S(N-l)$:
\begin{align}
    I^{\ot N-l}=\bigoplus_\alpha P_\alpha,
\end{align}
where $\alpha$ are all irreps of $S(N-l)$ which are present in the representation 
that permutes the "free" systems.
Now, if we apply some permutation $\sigma$ from $S(N)$, then 
we see that $P_\alpha$ is put to different systems, but still it becomes $P_\alpha$ as it is invariant under permutations. 
The same happens to any other constituent of $\rho$. 
Thus we can split $\rho$ into a direct sum of $\rho_\alpha$'s which instead of identity will have $P_\alpha$:
\begin{align}
    \rho=\bigoplus_\alpha \rho_\alpha
\end{align}
with 
\begin{align}
\rho_\alpha=\frac{1}{(N-l)!}\sum_{\sigma\in S(N)}V(\sigma) P_\alpha \ot P_+^{\otimes l} V(\sigma^{-1}). 
\end{align}
Thus $\rho$ is block diagonal, and the blocks are subspaces ${\cal H}_\alpha$, that are spanned by vectors of the form 
\begin{align}
   V(\sigma) |\phi_i^\alpha\>|\phi_+\>^{\ot l},
\end{align}
where $i=1,\ldots,d_\alpha$ with $d_\alpha$ being dimension of irrep $\alpha$; $\sigma\in S(N)$. 
However, in the above equation
we have much more vectors than needed to span ${\cal H}_\alpha$.
Indeed, a permutation $\sigma$
acting on $|\phi_i^\alpha\>|\phi_+\>^{\ot l}$ 
does the following: it changes the configuration of legs and applies some permutation on the free 
systems, see Fig. \ref{fig:induced}. 

\begin{figure}[h]
\includegraphics[width=.8\linewidth]{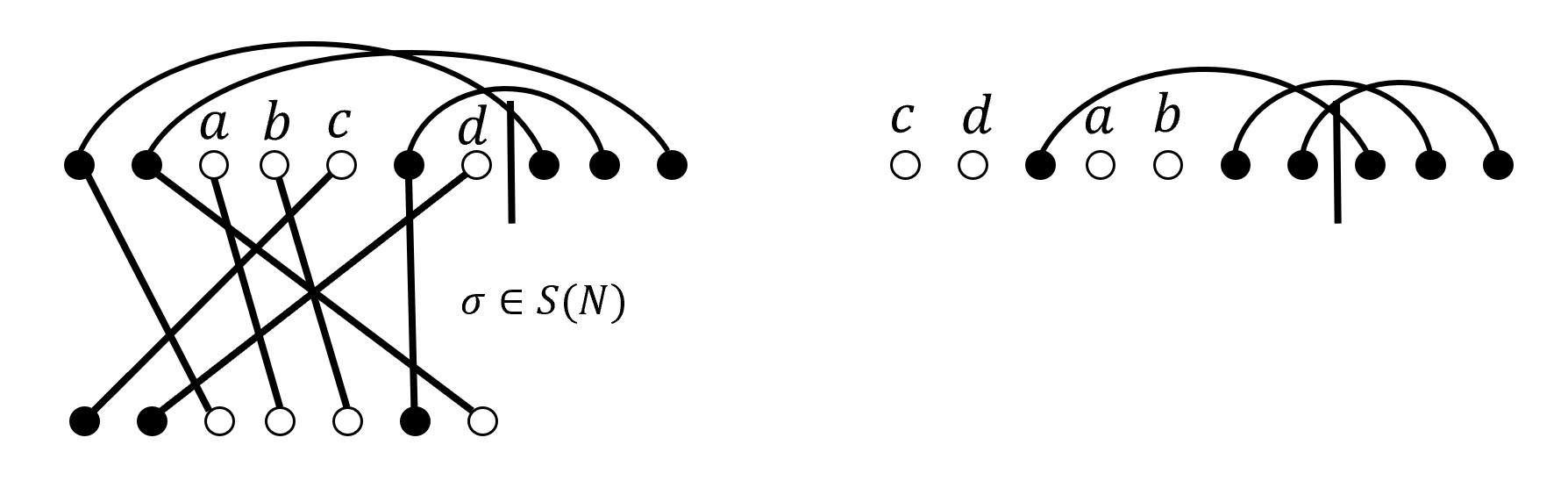}
	\caption{How permutation from $S(N)$ acts on exemplary constituent of $\rho$: it re-configures legs.}
\label{fig:induced}
\end{figure}
We see, that if we have some permutation $\sigma$ then all other permutations $\sigma'$ which have the same permutation on free systems are not needed - as the resulting vectors will have the same configuration of legs, and will differ only by permutation on free systems, which will enlarge dimension. Indeed on free systems, we have anyway full irrep basis, and permuting it we shall not go outside of the space. 

Thus since we are not interested in action instead of all $\sigma$ from $S(N)$ it is enough to take only ones from transversal, the quotient $S(N)/S(N-l)$, which we shall denote by $t_k, k=1,\ldots, N!/(N-l)!$:
\begin{align}
V(t_k)|\phi^\alpha_i\> |\phi_+\>^{\ot l} \equiv |t_k,\phi_i^\alpha\>.
\end{align} 
We are now in a position to examine the action of our group $S(N)$ on the space ${\cal H}_\alpha$. For $\sigma\in S(N)$ We write
\begin{align}
    V(\sigma) |t_k,\phi_i^\alpha\>=
    V(\sigma) V(t_k) |\phi_i^\alpha\>|\phi_+\>^{\ot l}.
\end{align}
Now we use that there is unique element $t_p$ of transversal, which satisfies:
\begin{align}
    \sigma t_k = t_p h, 
\end{align}
where $h\in S(N-l)$, so that 
\begin{align}
    V(\sigma) V(t_k) |\phi_i^\alpha\>|\phi_+\>^{\ot l}=
    V(t_p) V(h)|\phi_i^\alpha\>|\phi_+\>^{\ot l}=\sum_{j} \varphi^\alpha_{ji}(h)
    V(t_p)|\phi_j^\alpha\> |
\phi_+\>^{\ot l}
    =\sum_j \varphi^\alpha_{ji}(h) |t_p,\phi_j^\alpha\>.
\end{align}
However this is nothing but one of the definitions of the induced representation: it permutes elements of transversal, and applies an element from the subgroup to 
the irrep $\alpha$ (cf. \eqref{eq:induced})

Thus we have seen that if we apply permutation $\sigma\in S(N)$ to the total system, on the space ${\cal H}_\alpha$ it acts as the induced representation by irrep $\alpha$. 
Such representation is reducible, and the Frobenius reciprocity theorem says that the multiplicity of each irrep $\mu$ is equal to the multiplicity of the so-called, reduced representation $\alpha$ in $\mu$ (see below for definition of reduced representation).

\subsection{PRIR and proof of Frobenius theorem for symmetric group}
Suppose now that we do not know the Frobenius theorem, and let us proceed with setting orthonormal basis in $\mathcal{H}_\alpha$. 
Such basis was derived in \cite{StuIEEE} in the form of matrix basis, on space $\mathcal{H}_\alpha$ including multiplicity,  which is given by $m_\alpha$ - the multiplicity of irrep $\alpha$ 
in the natural representation of the symmetric group in $N-l$ systems (by natural we mean that it is given by operators $V(\sigma)$ which permute the systems, i.e. this is Schur-Weyl setup). 
To write it down we need a bit of preparation. 
First, for irrep $\mu$ of $S(N)$
we define the natural representation of the matrix basis (or matrix units) for the irrep $\mu$ including multiplicities as follows 
\begin{align}
    \label{eq:E}
    E_{ij}^\mu=
    \frac{d_{\mu}}{N!}\sum_{\sigma\in S(N)} \varphi^\mu_{ji}(\sigma^{-1})V(\sigma).
\end{align}
Here $\varphi^\mu_{ji}(\sigma^{-1})$ are matrix elements of irrep of $\sigma^{-1}\in S(N)$ and $i,j=1,\ldots, d_{\mu}$, where $d_{\mu}$ is dimension of the irrep $\mu$. 
Now consider subgroup $S(N-l)\subset S(N)$. If in our irrep $\mu$ we restrict to this subgroup, we obtain representation called {\it reduced representation}. 
It is reducible, and let $\alpha$ labels its irreducible components within irrep $\mu$,
determining blocks. Further, each copy of irrep $\alpha$ can appear with multiplicity, which we shall denote by $m_{\alpha/\mu}$. 
Now we can always assume that 
the basis in which the matrix elements of $\varphi^\mu(\sigma)$ are written are chosen in such a way
that the basis of \eqref{eq:E} is compatible with the $\alpha$ blocks.  This is what we call partially reduced irreducible representation.
Now the label $i_\mu$ can be alternatively written by means of a triple $i_\mu=(\alpha,r_{\mu/\alpha}, i_\alpha$). 
Namely, the index $\alpha$ tells in which irrep $\alpha$
we are, the index $r_{\mu/\alpha}$ tells in which copy of irrep $\alpha$ we are, and finally, 
the index $i_\alpha$
denotes the position in that copy 
(in the same way as $i_\mu$ denotes 
position in irrep $\mu$). 
One can use even shorter notation 
$i_\mu=(r_{\mu/\alpha},i_\alpha)$. 
We are now prepared to write 
the promised basis 
that appeared earlier in Theorem 11~\cite{StuIEEE}
\begin{align}
\label{eq:Fbasis}
F^{r_{\mu/\alpha} r_{\nu/\alpha}}_{i_{\mu}\quad j_{\nu}}:= d^k \frac{m_{\alpha}}{\sqrt{m_{\mu}m_{\nu}}}E_{i_{\mu} \ 1_{\alpha}}^{\quad r_{\mu/\alpha}} P_+^{\ot k}E^{r_{\nu/\alpha}}_{1_{\alpha}\quad  j_{\nu}}
\end{align} 
where $m_\alpha$ 
is multiplicity of irrep $\alpha$ in natural representation of $S(N-l)$ and 
and $m_\mu$, $m_\nu$ are multiplicities  of natural representation of $S(N)$, and we have written one of the index of  the operators $E$ in the partially reduced notation:
\begin{align}
E_{i_{\mu} \ 1_{\alpha}}^{\quad r_{\mu/\alpha}}\equiv
E^\mu_{i_{\mu}i'_\mu}
\end{align}
with 
$i'_\mu=(r_{\mu/\alpha},1_\alpha)$. Here $1_\alpha$ represents some arbitrarily fixed label - the operator does not depend on its choice.
Now \cite{StuIEEE}
the operator 
\begin{align}
    \sum_{i_\mu}
    F^{r_{\mu/\alpha} r_{\mu/\alpha}}_{i_{\mu}\quad i_{\mu}}
\end{align}
projects onto irrep $\mu$ within $H_\alpha$, which is the space for representation of $S(N)$ induced by irrep $\alpha$ of its subgroup $S(N-l)$. 
Thus we see here explicitly, how the indices $r$ that originally count multiplicity of $\alpha$ within reduced representation of $\mu$,
now count multiplicity of $\mu$ within the induced representation. Equation~\eqref{eq:Fbasis} is a special case of the more general result presented in Theorem 1.4 in~\cite{Ram1992MatrixUF}.

This is just the contents of 
the Frobenius reciprocity theorem 
in the special case for the symmetric group. In the following, inspired by this example, we shall prove the Frobenius reciprocity theorem for general group and subgroup.

\subsection{Duality between induced and reduced representation via PBT operator versus Jucys-Murphy operator for $l=1$}
\label{subsec:Jucys}
The dualism of reduced and induced representation manifested by the Frobenius theorem is indeed very directly seen in our above derivation of the theorem. In this section, we shall present this dualism in yet another way. 
To do so let us define the Jucys-Murphy operator~\cite{Mu,Ju,vershikBook}, which in the natural representation, is given by  
\begin{align}
J_N=\sum_{a=1}^{N-1}  V((a,N)),
\end{align}
where $V((a,N))$ are as before operators that permute $N$ systems according to permutation $(a,N)$. It is clear from its definition that the operator $J_N$ belongs to the matrix algebra $\mathcal{A}_N(d)$ defined through~\eqref{eq:groupAlg}. Due to the discussion presented in Section~\ref{secBrauer}, the operator $J_N$ is the matrix representation of the element from the Brauer algebra $\mathcal{B}^d_{N}$. The Jucys-Murphy elements play an important role in the representation theory of the symmetric group. In principle, they have found an application to an alternative approach in the construction of irreducible representations of $S(n)$ called Okounkov-Vershik approach~\cite{Okunkov,Vershik2005,vershikBook}. The spectrum of the Jucys-Murphy elements is known and discussed for example in Section 4.5 of~\cite{grinko2023gelfandtsetlin}.
Now let us consider the operator $J_{N+1}$ and apply to it partial transpose on $(N+1)-$th system. Now, since the partially transposed swap operator is equal to $dP_+$
we see that  PBT operator $\rho$
is just a partially transposed Jucys-Murphy operator on $N+1$ systems
\begin{align}
\label{eq:rhoToJM}
\rho=(J_{N+1})^{T_{N+1}}.
\end{align}
Clearly, the above operator belongs to the algebra of partially transposed permutation operators $\mathcal{A}_{N+1}^{T_{N+1}}(d)$ defined through Definition~\ref{def_A}. In other words, the operator $\rho$ is the matrix representation of the element from the walled Brauer algebra $\mathcal{B}^d_{N,1}$, according to the discussion presented in Section~\ref{secBrauer}.
The eigenvalues of this operator were found in 
\cite{Studzinski2017}  and are given by 
\begin{align}
\lambda_\mu(\alpha)=N\frac{m_\mu d_\alpha}{m_\alpha d_\mu}
\end{align}
and multiplicity of the eigenvalue is $d_\mu m_\alpha$. 
In the rest of the paper, we shall find eigenvalues of ordinary Jucys-Murphy operator $J_N$
(see Corollary \ref{cor:jucys-eig}).
They are given by 
\begin{align}
\gamma_\mu(\alpha)
=N\frac{m_\mu d_\alpha}{m_\alpha d_\mu} - d
\end{align}
and multiplicity of the eigenvalue 
is $m_\mu d_\alpha$. 
We have summarised this in Table~\ref{tab:jucys}
    \begin{table}[t]
\centering
\begin{tabular}
{||l|c|c||}
\hline
operator & eigenvalue &multiplicity \\ \hline
$J_N$ & $N\frac{m_\mu d_\alpha}{m_\alpha d_\mu} -d$ & $m_\mu d_\alpha$\\ \hline
$(J_{N+1})^{T_{N+1}}$ &$N\frac{m_\mu d_\alpha}{m_\alpha d_\mu}$ &$m_\alpha d_\mu$\\ 
\hline
\end{tabular}
\caption{Table collects relations between the eigenvalues of Jucys-Murphy operator $J_N$ and partially transposed Jucys-Murphy operator $(J_{N+1})^{T_{N+1}}$, or equivalently the PBT operator $\rho$, see~\eqref{eq:rhoToJM}. We see that the spectrum of $\rho$ is just shifted by $d$ spectrum of the Jucys-Murphy element $J_N$. Notice that operators $\rho$ and $J_N$ are defined for the different number of systems, i.e. they differ by one system. Here $\mu$ labels irreps of $S(N)$ that can be induced from irreps $\alpha$ of $S(N-1)$.}
\label{tab:jucys}
\end{table}
The discussion in the previous section shows how the multiplicity for transposed Jucys-Murphy element $\rho=(J_{N+1})^{T_{N+1}}$ is related to induced representation. Namely,
$m_\alpha$ term in multiplicity comes from the fact that space ${\cal H}_\alpha $ is labeled by $\alpha$,  which repeats $m_\alpha$ times. 
And the multiplicity $m_\mu$ comes from the fact, that $\rho$ is invariant under $\mu$, hence by Schur lemma, is proportional to identity on irreps of $\mu$. 

Similarly, it is easy to understand the opposite formula for the multiplicity of 
the Jucys-Murphy operator $J_N$.
First, since it is a combination of 
permutations, it is a direct sum over irreps $\mu$ with multiplicity $m_\mu$. Next, within irrep $\mu$, 
it is constant on irrep $\alpha$ of subgroup $S(N-1)$, hence  we have term  $d_\alpha$. 


\section{The concept of $PRIR$ for an arbitrary group $G$ and its
subgroup $H\subset G$.}
\label{Sec:PRIR}
Let $H\subset G$ be an arbitrary subgroup of $G$ with transversal $%
T=\{t_{k}:k=1,\ldots,\frac{|G|}{|H|}\equiv s\}$, i.e. we have a coset decomposition
\begin{equation}
\label{Sec2eq13}
G=\bigcup _{k=1}^{s}t_{k}H, \quad \quad \forall g\in G\quad g=t_{k}h:h\in H,
\end{equation}
where the last decomposition of $g\in G$ is unique.
Now, let us consider an arbitrary unitary irreducible representation (irrep) $\psi ^{\mu }$ of $G$, which
also will be denoted briefly as $\mu \in \widehat{G}$ where the latter is the set
of all irreps of the group $G$. The irrep $\psi ^{\mu }$ can be always unitarily
transformed to a partially reduced form, such that 
\begin{equation}
\label{S2eq2}
\operatorname{Res}\downarrow _{H}^{G}(\psi _{R}^{\mu })=\bigoplus _{\alpha \in \mu
,a_{\alpha }}\varphi ^{\alpha (a_{\alpha })},\quad a_{\alpha
}=1,\ldots,m_{\alpha }^{\mu },
\end{equation}%
where $\alpha $ labels the type of irrep of $H$ and $a_{\alpha }$ is
a number of irrep type $\alpha $ in the decomposition~\eqref{S2eq2}, and $m_{\alpha
}^{\mu }$ is the multiplicity of irrep $\varphi ^{\alpha }$ in $\operatorname{Res}\downarrow _{H}^{G}(\phi _{R}^{\mu })$. Decomposition~\eqref{S2eq2} is not unique. We
will assume that irreps $\varphi ^{\beta (b_{\beta })}$ of subgroup $H$
are identical, i.e. we have 
\begin{equation}
\varphi ^{\alpha (a_{\alpha })}=\varphi ^{\alpha },\quad \forall a_{\alpha
}=1,\ldots,m_{\alpha }^{\mu }.
\end{equation}

\begin{definition}
\bigskip A unitary matrix irrep $\psi ^{\mu }$ of $G$, with partially
reduced form~\eqref{S2eq2}, i.e. which on the subgroup $H$ has block diagonal form, we call Partially Reduced Irreducible Representation $(PRIR).$
\end{definition}

From the above, it follows that the diagonal blocks in the decomposition~\eqref{S2eq2} are labeled and in fact
ordered by the index $\alpha (a_{\alpha })$ and inside diagonal blocks the
matrix elements are labelled by indices $j_{\alpha }=1,\ldots,\dim \varphi
^{\alpha }=d_{\alpha }$ of irrep $\varphi^{\alpha (a_{\alpha })}$ of the
subgroup $H$, which is included in irrep $\psi^{\mu }$ of $G$, so in
matrix notation we have%
\begin{equation}
\label{Sec2eq3X}
\forall h\in H\quad (\psi _{R}^{\mu })_{i_{\alpha }\quad \quad j_{\beta
}}^{\alpha (a_{\alpha })\beta (b_{\beta })}(h)=\delta ^{\alpha \beta }\delta
^{a_{\alpha }b_{\beta }}\varphi _{i_{\alpha }j_{\alpha }}^{\alpha (a_{\alpha
})}(h).
\end{equation}%

The block structure of this reduced representation allows us to introduce
such a block indexation for irrep $\psi ^{\mu }$ for all elements of $G$%
\begin{equation}
\forall g\in G\quad \psi _{R}^{\mu }(g)=((\psi _{R}^{\mu })_{k_{\mu }l_{\mu
}}(g))=\bigl((\psi _{R}^{\mu })_{i_{\alpha }\quad \quad j_{\beta }}^{\alpha
(a_{\alpha })\beta (b_{\beta })}(g)\bigr),
\end{equation}%
where the indices $k_{\mu },l_{\mu }$ are standard matrix indices, so in $%
PRIR$ the standard matrix indices are replaced by indexation directly
connected with irreps of subgroup $H$ included in irrep $\psi
^{\mu }$ of $G$. Note that the diagonal blocks are square, whereas the
off-diagonal blocks in general need not be square. Thus we see that $%
PRIR^{\prime }s$ have two main features:
\begin{enumerate}
\item They are partially reduced on subgroup $H$.
\item The matrix elements of irrep $\psi ^{\mu }$ of $%
G $ are labelled by multi-indices $\binom{\alpha (a_{\alpha })}{i_{\alpha }}%
, $ $\binom{\beta (b_{\beta })}{j_{\beta }}.$
\end{enumerate}

Note that in definition of $PRIR$ we do not assume that $\operatorname{Res}
\downarrow _{H}^{G}(\psi ^{\mu })$ is simply reducible, which is important
because it allows us to give a new proof of the Frobenius reciprocity theorem
(see Thm.~\ref{Sec2Thm18}). 
Actually, if we have an inclusion chain of subgroups 
so that each inclusion is multiplicity-free, 
then we obtain a chain of $PRIRs$ 
which are multiplicity-free, which leads to the well-known Gelfand-Tsetlin basis. As said, we are interested in $PRIR$ that is not necessarily multiplicity-free. As a matter of fact, $PRIR$ is strictly related to subgroup adapted basis~\cite{Koch_2012}. 
The latter is the basis in 
the irreducible representation of a group whose elements are basis vectors from 
individual irreps of the subgroup (which are subspaces of the irrep of the group).
Now $PRIR$ is simply irreducible representation of the group written in subgroup-adapted basis.
On the level of matrices,
this is equivalent to saying that $PRIRs$ are 
block diagonal on the subgroup as in eq.~\eqref{S2eq2}.

The above-introduced indexation of $PRIRs$ is more complicated than
the standard one  but, for example, it allows to derive
some important new relations among matrix elements of $PRIRs$
which would be difficult to rewrite in standard indexation - see Proposition~\ref{Sec2Prop2} below).

\bigskip Using equation~\eqref{Sec2eq3X} we get 
\begin{equation}
\label{Sec2eq17}
(\psi _{R}^{\mu })_{i_{\alpha }\quad \quad j_{\beta }}^{\alpha (a_{\alpha
})\beta (b_{\beta })}(t_{k}h)=\sum_{k_{\beta }}(\psi _{R}^{\mu })_{i_{\alpha
}\quad \quad k_{\beta }}^{\alpha (a_{\alpha })\beta (b_{\beta
})}(t_{k})\varphi _{k_{\beta }j_{\beta }}^{\beta (b_{\beta })}(h),\quad
\forall h\in H\quad \forall t_{k}\in T. 
\end{equation}
From this, it follows that multiplication by a matrix representing elements of subgroup $H$ is
simpler than in non-$PRIR$ representations.

Now we may formulate the first important property of $PRIRs$, which is
a kind of orthogonality relation for matrix elements of transversal $T$.

\begin{proposition}
\label{Sec2Prop2}
Suppose that 
\begin{equation}
\varphi ^{\beta }\in \operatorname{Res}\downarrow _{H}^{G}(\psi _{R}^{\mu })\quad
and\quad \varphi ^{\beta }\in \operatorname{Res}\downarrow _{H}^{G}(\psi _{R}^{\nu }).
\end{equation}%
Then the matrix elements of $PRIR^{\prime }s$ $\psi
_{R}^{\mu }$ and $\phi _{R}^{\nu }$ of $G$ satisfy the following sum rule%
\begin{equation}
\sum_{k=1}^{s}\sum_{k_{\beta }=1}^{|\beta |}(\psi _{R}^{\mu })_{i_{\alpha
}\quad \quad k_{\beta }}^{\alpha (a_{\alpha })\beta (b_{\beta
})}(t_{k}^{-1})(\psi _{R}^{\nu })_{k_{\beta }\quad \quad j_{\gamma
}}^{(\beta ,b_{\beta }^{\prime })(\gamma ,c_{\gamma })}(t_{k})=\frac{|G|}{|H|%
}\frac{d_{\beta }}{d_{\mu }}\delta ^{\mu \nu }\delta ^{b_{\beta }b_{\beta
}^{\prime }}\delta ^{\alpha \gamma }\delta ^{a_{\alpha }c_{\gamma }}\delta
_{i_{\alpha }j_{\gamma }},
\end{equation}%
where $\varphi ^{\alpha }\in \operatorname{Res}\downarrow _{H}^{G}(\phi _{R}^{\mu
})\quad and\quad \varphi ^{\gamma }\in \operatorname{Res}\downarrow _{H}^{G}(\phi
_{R}^{\nu })$.
\end{proposition}

\begin{proof}
The proof is based on the classical orthogonality relations for irreps, which in $PRIR$ notation takes the form
\begin{equation}
\label{Sec2eq3}
\sum_{g\in G}(\psi _{R}^{\mu })_{i_{\alpha }\quad \quad k_{\beta }}^{\alpha
(a_{\alpha })\beta (b_{\beta })}(g^{-1})(\psi _{R}^{\nu })_{k_{\beta }\quad
\quad j_{\gamma }}^{\beta (b_{\beta }^{\prime })\gamma (c_{\gamma })}(g)=%
\frac{|G|}{d_{\mu }}\delta ^{\mu \nu }\delta ^{b_{\beta }b_{\beta }^{\prime
}}\delta ^{\alpha \gamma }\delta ^{a_{\alpha }c_{\gamma }}\delta _{i_{\alpha
}j_{\gamma }},
\end{equation}%
which means, that even if $\alpha =\gamma $, i.e. these representations
are of the same type, but $a_{\alpha }\neq c_{\gamma }$, the $RHS$ of~\eqref{Sec2eq3} is equal to zero. 
Similarly, if the indices $b_{\beta},b_{\beta }^{\prime }$ which enumerate irreps $\beta $ are not
equal, then $RHS$ of~\eqref{Sec2eq3} is equal to zero. This follows from the fact that in the
classical orthogonality relations if corresponding indices of irreps are not equal, then the $RHS$ of the orthogonality relations is equal to zero, and in $PRIR$ notation the irreps $\varphi ^{\alpha }$
of subgroup $H$ and their indices $i_{\alpha }$ play the role of indices in irrepss $\psi _{R}^{\mu }$. The next part of the proof is a simple generalization of the proof of Proposition 17 in~\cite{MozJPA}.
\end{proof}

\begin{remark}
From the form of $LHS$ of equation~\eqref{Sec2eq3} in the thesis of Proposition~\ref{Sec2Prop2},
in particular from the fact that the second sum in $LHS$ runs over a lower part
of multi-index $\binom{\beta (b_{\beta })}{k_{\beta }}$ 
only (over $k_{\beta }$) it is clear that $PRIR$ multi-index notation is essential to formulate
and prove this result.
\end{remark}

Next result concerns properties of the transversal $T=\{t_{k}:k=1,\ldots,\frac{
|G|}{|H|}\equiv s\}$ of the subgroup $H\subset G$ and transversal element $\Upsilon =\sum_{k=1}^{s}t_{k}$ in the group algebra $\mathbb{C}[G].$ It is well known that any irrep of any group $G$ is also an irrep  the corresponding group algebra $\mathbb{C}[G]$, so we may use the concept of $PRIR^{\prime }s$ in the group
algebra.

\begin{proposition}
\label{Sec2Prop4}
Suppose that $\{t_k\}$ is any transversal of a group $G$ with respect to subgroup $H$~\eqref{Sec2eq13}, $\psi _{R}^{\mu }$ is $PRIR$ of a group
algebra $\mathbb{C}[G]$, such that 
\begin{enumerate}
    \item $\operatorname{Res}\downarrow _{H}^{G}(\psi _{R}^{\mu })$ is
simply reducible, 
    \item and we have
    \begin{equation}
\label{Sec2eq4}
\forall h\in H\quad h(\sum_{k=1}^{s}t_{k})h^{-1}=(\sum_{k=1}^{s}t_{k}),
\end{equation}%
\end{enumerate}
then 
\begin{equation}
\label{Sec2eq5}
\sum_{k=1}^{s}(\psi _{R}^{\mu })(t_{k})=\bigoplus _{\alpha \in \mu }\eta _{\mu
}(\alpha )\mathbf{1}_{\alpha },
\end{equation}%
i.e. on $RHS$ we have a block diagonal matrix, such that the diagonal blocks
are unit matrices multiplied by numbers $\eta _{\mu }(\alpha )$. These numbers
do not depend on the form of $PRIR$ $\psi _{R}^{\mu }$, i.e. for all
possible choices of $PRIR$ representation, satisfying the above assumptions,
these numbers are the same.
\end{proposition}

\begin{proof}
We give here a sketch of the proof. The assumption of simple reducibility
and condition~\eqref{Sec2eq4} together with classical Schur Lemma for irreps
implies that equation~\eqref{Sec2eq5} must hold. This result is easy to see due to
the $PRIR$ blocks indexation of $PRIR$ $\psi_{R}^{\mu }$.
\end{proof}

From Proposition~\ref{Sec2Prop4} we get immediately:
\begin{corollary}
For any $PRIR$ $\psi _{R}^{\mu }$ transversal element $\Upsilon
=\sum_{k=1}^{s}t_{k}\in \mathbb{C}[G]$ takes in $\psi _{R}^{\mu }$ diagonal form 
\begin{equation}
\psi _{R}^{\mu }(\Upsilon )=\bigoplus _{\alpha \in \mu }\eta _{\mu }(\alpha )%
\mathbf{1}_{\alpha }
\end{equation}
and the numbers $\eta _{\mu }(\alpha )$ form the spectrum of the operator $\psi
_{R}^{\mu }(\Upsilon )$. Multiplicity of eigenvalue $\eta _{\mu
}(\alpha )$ is equal to $\dim \varphi ^{\alpha }=d_{\alpha }$. Decomposing
arbitrary representation $\Psi $ of the group $G$ into $PRIR^{\prime }s$ $%
\psi _{R}^{\mu }$ and using Proposition~\ref{Sec2Prop4} we get the spectral
decomposition of the transversal operator $\Psi (\Upsilon )$.
\end{corollary}


As an important example of application of $PRIRs$ we have the following:

\begin{proposition}
Let $G=S(n)$ and $H=S(n-1)$ with standard transversal $T=\{(an):a=1,\ldots,n\}$,
which satisfy the assumptions of Proposition~\ref{Sec2Prop4}. Then for any $PRIR$ $\psi_{R}^{\mu }$ of the group $S(n)$ we have 
\begin{equation}
\sum_{a=1}^{n}(\psi _{R}^{\mu })(an)=\psi _{R}^{\mu }(\Upsilon )=\bigoplus
_{\alpha \in \mu }\eta _{\mu }(\alpha )\mathbf{1}_{\alpha },
\end{equation}
\begin{equation}
\eta _{\mu }(\alpha )=\left( 
\begin{array}{c}
n \\ 
2
\end{array}
\right) \frac{\chi ^{\mu }(12)}{d_{\mu }}-\left( 
\begin{array}{c}
n-1 \\ 
2
\end{array}
\right) \frac{\chi ^{\alpha }(12)}{d_{\alpha }}+1,\quad
\end{equation}
where $\chi ^{\mu }$ is the character of $\psi _{R}^{\mu }$ and $\chi
^{\alpha }$ s the character of $\varphi ^{\alpha }$. We see that numbers $%
\eta _{\mu }(\alpha )$ depend only on the pure characteristics of irreps $\psi _{R}^{\mu }$ and $\varphi ^{\alpha }\in \operatorname{Res}\downarrow _{H}^{G}(\phi _{R}^{\mu })$ (for brevity $\alpha \in \mu )$.
\end{proposition}

\begin{proof}
The starting point in this proof is the following equation%
\begin{equation}
\sum_{a<b}^{n}(\psi _{R}^{\mu })(ab)=\left( 
\begin{array}{c}
n \\ 
2
\end{array}
\right) \frac{\chi ^{\mu }(12)}{d_{\mu }}\mathbf{1}_{d_{\mu }},
\end{equation}
where $\mathbf{1}_{d_{\mu }}$ is the unit matrix of dimension $d_{\mu }$ and on $LHS$ we have the sum of all transpositions in $S(n)$ which forms an equivalence class. This equation follows from the fact
that the famous Schur lemma implies that the sum of all elements of any
the equivalence class of an arbitrary group in any irrep is proportional to the unit
matrix. Taking trace we derive the coefficient of proportionality. Next we
rewrite $LHS$ as follows 
\begin{equation}
\sum_{a<b}^{n}(\psi _{R}^{\mu })(ab)=\sum_{a=1}^{n-1}(\psi _{R}^{\mu
})(an)+\sum_{a<b}^{n-1}(\psi _{R}^{\mu })(ab),
\end{equation}%
so in the second sum on $RHS$  the summation is over all equivalence
class of transpositions but in the subgroup $S(n-1)$. Therefore, from
definition of $PRIRs$ matrix of the second sum is block diagonal
matrix and in the diagonal blocks we have irreps of $S(n-1)$
included in $\psi _{R}^{\mu }$. Again from the Schur lemma we get in each
diagonal block in the last sum 
\begin{equation}
\sum_{a<b}^{n}(\varphi ^{\alpha })(ab)=\left( 
\begin{array}{c}
n-1 \\ 
2
\end{array}
\right) \frac{\chi ^{\alpha }(12)}{d_{\alpha }}\mathbf{1}_{d_{\alpha }}
\end{equation}
for any irrep $\varphi ^{\alpha }$ of $S(n-1)$ included in $\psi _{R}^{\mu
}$, so we see that diagonal blocks in matrix $\sum_{a<b}^{n-1}(\psi
_{R}^{\mu })(ab)$ are also diagonal. Shifting the diagonal matrix $
\sum_{a<b}^{n-1}(\psi _{R}^{\mu })(ab)$ on $RHS$ of the first equation we
get the result.
\end{proof}
In the above proof, we see that the defining property of $PRIRs$, i.e. partial reduction plays an essential role in it.

\begin{example}
In the simplest case, i.e. for identity irrep $\mu =id$, we have
\begin{equation}
\eta _{id}(id)=n.
\end{equation}
\end{example}

In general the above formulae for eigenvalues $\eta _{\mu }(\alpha )$ of
the transversal operator $\Upsilon $ are not entirely analytical because for
natural characteristics of irreps characters $\chi^{\mu}$
and corresponding dimension $d_{\mu }$ there are not analytical expressions.
However, for particular irreps of group $S(n)$ there exists
purely analytical formulae for these quantities. Namely, we have 

\begin{lemma}[see~\cite{Fulton1991-book-rep}] 
\label{Lemma:Harris}
Let $\mu \equiv \mu _{k}=(n-k,k):1\leq
k\leq \frac{1}{2}n$ be a two-row  partition of irrep of $S(n)$, then 
\begin{equation}
\label{eq:FH}
\chi ^{\mu _{k}}(12)=\binom{n-2}{k}+\binom{n-2}{k-2}-\binom{n-2}{k-1}-\binom{n-2
}{k-3},
\end{equation}
and 
\begin{equation}
\dim \mu _{k}=d_{\mu _{k}}=\binom{n}{k}-\binom{n}{k-1}.
\end{equation}
When $k=1$ we set $\binom{n-2}{k-2}\equiv 0\equiv \binom{n-2}{k-3}$, and when $k=2$ we set $\binom{n-2}{k-3}\equiv 0$ in~\eqref{eq:FH}.
\end{lemma}

Next, we have for $n-k>k$
\begin{equation}
\label{eq:decomp1}
\operatorname{Res}\downarrow _{S(n-1)}^{S(n)}(\psi^{\mu
_{k}})= \varphi^{\alpha _{k}}\oplus \varphi^{\alpha _{k-1}},
\end{equation}
where $\alpha_k=(n-1-k,k),$ and $\alpha_{k-1}=(n-1-(k-1),k-1)$,
so $\alpha _{k}, \alpha _{k-1}$ are irreps of $S(n-1)$. For $n-k=k$, we have
\begin{equation}
\label{eq:decomp2}
\operatorname{Res}\downarrow _{S(n-1)}^{S(n)}(\psi^{\mu_{k}})=\varphi^{\alpha _{k-1}},
\end{equation}
where $\alpha _{k-1}=(n-1-(k-1),k-1)$.

Using Lemma~\ref{Lemma:Harris} and expressions~\eqref{eq:decomp1},~\eqref{eq:decomp2}, we derive purely analytical formulae for the eigenvalues of the transversal operator $\Upsilon $ in two rows $PRIRs$ of the group $S(n)$.

\begin{proposition}
\label{L:JMd2}
Let $\mu \equiv \mu _{k}=(n-k,k):1\leq k\leq \frac{1}{2}n$ be a two row
partition of irrep of $S(n)$, then 
\begin{enumerate}
\item If $n-k>k$ the transversal operator $\Upsilon =\sum_{a=1,\ldots,n}(an)$ takes in $PRIR$ $\psi^{\mu _{k}}$ diagonal form with two different eigenvalues 
\begin{equation}
\eta _{\mu _{k}}(\alpha _{k})=\binom{n}{2}\frac{\binom{n-2}{k}+\binom{n-2}{
k-2}-\binom{n-2}{k-1}-\binom{n-2}{k-3}}{\binom{n}{k}-\binom{n}{k-1}}-\binom{n-1}{2}\frac{\binom{n-3}{k}+\binom{n-3}{k-2}-\binom{n-3}{k-1}-\binom{
n-3}{k-3}}{\binom{n-1}{k}-\binom{n-1}{k-1}}+1
\end{equation}
with multiplicity $d_{\alpha _{k}}=\binom{n-1}{k}-\binom{n-1}{k-1}$, and 
\begin{equation}
\eta _{\mu _{k}}(\alpha _{k-1})=\binom{n}{2}\frac{\binom{n-2}{k}+\binom{n-2}{
k-2}-\binom{n-2}{k-1}-\binom{n-2}{k-3}}{\binom{n}{k}-\binom{n}{k-1}}-\binom{n-1}{2}\frac{\binom{n-3}{k-1}+\binom{n-3}{k-3}-\binom{n-3}{k-2}-
\binom{n-3}{k-4}}{\binom{n-1}{k-1}-\binom{n-1}{k-2}}+1
\end{equation}
with multiplicity $d_{\alpha _{k-1}}=\binom{n-1}{k-1}-\binom{n-1}{k-2}$.

\item If $n-k=k$ then $\psi ^{\mu _{k}}(\Upsilon )$ has only one eigenvalue
equal to  $\eta _{\mu _{k}}(\alpha _{k-1})$ with multiplicity $d_{\alpha
_{k-1}}=\binom{n-1}{k-1}-\binom{n-1}{k-2}=d_{\mu _{k}}.$ In this particular case, $\psi ^{\mu _{k}}(\Upsilon )$ is proportional to unit matrix.
\end{enumerate}
\end{proposition}

\begin{definition}
The element $J_{n}=\sum_{a=1}^{n-1}(an)$ is called Jucys-Murphy
element (JM element) in the group algebra $\mathbb{C}[S(n)]$ and together with reduced JM element $
J_{k}=\sum_{a=1}^{k-1}(ak),$ $k=2,\ldots,n$ it plays an important role in
the representation theory of symmetric group $S(n)$~\cite{Mu, Ju, Okunkov,Tullio}.
\end{definition}

In general, calculating analytically the spectrum of JM elements 
in the group algebra $\mathbb{C}[S(n)]$ is not easy (see~\cite{Tullio}). Using Proposition~\ref{Sec2Prop4} together with the above example we get

\begin{corollary}
\label{Sec2Cor2}
In any $PRIR$ $\psi _{R}^{\mu }$ of the symmetric group $S(n)$ we have the
following spectral decomposition JM operator $J_{n}=
\sum_{a=1}^{n-1}(an)$ 
\begin{equation}
\psi _{R}^{\mu }(J_{n})=\bigoplus _{\alpha \in \mu }\gamma _{\mu }(\alpha )
\mathbf{1}_{\alpha },
\end{equation}
where 
\begin{equation}
\gamma _{\mu }(\alpha )=\left( 
\begin{array}{c}
n \\ 
2
\end{array}
\right) \frac{\chi ^{\mu }(12)}{d_{\mu }}-\left( 
\begin{array}{c}
n-1 \\ 
2
\end{array}
\right) \frac{\chi ^{\alpha }(12)}{d_{\alpha }}=\eta _{\mu }(\alpha )-1,
\end{equation}
and multiplicity of the eigenvalue $\gamma _{\mu }(\alpha )$ is equal to $
d_{\alpha }$.
\end{corollary}
Again, knowing the spectrum of the operator $J_{n}$ in any $PRIR$, one can
determine the spectrum of $J_{n}$ in any representation.
As an example let us consider the natural permutation representation of
symmetric as it is defined in~\eqref{eq:natrep}.
It is known~\cite{Fulton1991-book-rep} that then for $d=2$, we have the following  
\begin{equation}
V_{d=2}^{n}=m_{id}\psi ^{id}\oplus \left[ \bigoplus_{1\leq
k\leq \frac{1}{2}n}m_{\mu _{k}}\psi ^{\mu _{k}}\right],
\end{equation}%
where $m_{id}=n+1$, and $m_{\mu _{k}}=n-2k+1$ are the multiplicities of
corresponding irreps which are assumed to be $PRIRs$. Now using the above results we may derive the spectrum of the $JM$ operator in the natural representation $V_{d=2}^{n}$.

\begin{proposition}
The spectrum of the operator $V_{d=2}^{n}(J_{n})$ acting in the space $(\mathbb{C}^{2})^{\otimes n}$ is the following
\begin{enumerate}
\item If $\mu_k=id$, then
\begin{equation}
\gamma _{id}(id)=n-1,
\end{equation}
with multiplicity $n+1$. 
\item If $\mu _{k}$ is such that $n-k>k$, then   
\begin{equation}
\gamma _{\mu _{k}}(\alpha _{k})=\eta _{\mu _{k}}(\alpha _{k})-1,
\end{equation}
with multiplicity $m_{\mu _{k}}d_{\alpha _{k}}=(n-2k+1)\binom{n-1}{k}-\binom{n-1}{k-1}$, and 
\begin{equation}
\gamma _{\mu _{k}}(\alpha _{k-1})=\eta _{\mu _{k}}(\alpha _{k-1})-1,
\end{equation}
with multiplicity $m_{\mu _{k}}d_{\alpha _{k-1}}=(n-2k+1)\binom{n-1}{k-1}-\binom{n-1}{k-2}$.
\item If $\mu _{k}$ is such that $n-k=k$, then 
\begin{equation}
\gamma _{\mu _{k}}(\alpha _{k-1})=\eta _{\mu _{k}}(\alpha _{k-1})-1
\end{equation}
with multiplicity $m_{\mu _{k}}d_{\alpha _{k-1}}.$
\end{enumerate}
\end{proposition}

It is well known that any group algebra $\mathbb{C}[G]$ ($G$ is finite) has the following decomposition into irreps
\begin{equation}
\mathbb{C}[G]=\bigoplus _{\mu \in \widehat{G}}m_{\mu }\psi ^{\mu },\quad 
\end{equation}%
where $m_{\mu }=\dim \psi ^{\mu }$. From this and Corollary~\ref{Sec2Cor2} we get:
\begin{proposition}
The spectrum of JM element $J_{n}=\sum_{a=1}^{n-1}(an)$, acting
on the group algebra $\mathbb{C}[S(n)]$, i.e in regular representation is given by numbers $\gamma
_{\mu }(\alpha )$ and multiplicity of eigenvalue $\gamma _{\mu }(\alpha )$
is equal to $m_{\mu }d_{\alpha }$.
\end{proposition}

We see that by using the $PRIR$ concept one can express the eigenvalues of the 
JM element $J_{n}=\sum_{a=1}^{n-1}(an)$ via characters of irreps of the
groups $S(n)$ and $S(n-1)$, which are basic characteristics of irreps. Their multiplicities are also expressed by basic group representation parameters. Obtained expressions for the spectrum of the JM elements are of a different nature than in~\cite{Okunkov,Vershik2005,Tullio}.

\bigskip One very efficient application of $PRIR$ concept was dedicated studies 
on  PBT-operator in deterministic port-based teleportation scheme, see~\eqref{eq1} and~\eqref{eq2}.  Hereunder we use original notation and we denote the total number of systems (number of ports + teleported state) by $n$, while the number of ports by $N$, and we have of course $n=N+1$. The algebraic structure of the port-based teleportation scheme is
connected with the algebra of partially transposed operators $\mathcal{A}_{n}^{t_{n}}(d)$
acting in $n$-fold tensor product of $d-$dimensional vector space.  This
algebra is not isomorphic (except when $d=2$) with standard permutational
representation of symmetric group $S(n)$ in the Schur-Weyl construction. Due to
the application of properties of $PRIR$'s it was possible to derive elegant
expression for the spectrum of PBT-operator and to derive entanglement fidelity for
deterministic PBT-scheme in all variants. The result for the mentioned eigenvalues are contained in~\cite{Studzinski2017,MozJPA}.

\begin{proposition}
\label{Sec2Prop10}
The eigenvalues $\lambda _{\mu }(\alpha )$ of PBT-operator with $N=n-1$ ports are of the following form 
\begin{equation}
\label{eq:55}
\lambda _{\mu }(\alpha )=(n-1)\frac{m_{\mu }d_{\alpha }}{m_{\alpha }d_{\mu }}%
=\left( 
\begin{array}{c}
n-1 \\ 
2%
\end{array}%
\right) \frac{\chi ^{\mu }(12)}{d_{\mu }}-\left( 
\begin{array}{c}
n-2 \\ 
2%
\end{array}%
\right) \frac{\chi ^{\alpha }(12)}{d_{\alpha }}+d,
\end{equation}%
where $d_{\alpha },d_{\mu }$ are dimensions of irreps $\psi
^{\mu }\in \widehat{S(n-1)}$, $\varphi ^{\alpha }\in \widehat{S(n-2)}$ and $%
m_{\mu },m_{\alpha }$ are corresponding multiplicities is the standard swap
representation of $S(n-1)$ and $S(n-1)$ respectively.
\end{proposition}
Comparing Proposition~\ref{Sec2Prop10} with Corollary~\ref{Sec2Cor2} we get
\begin{corollary}
\label{cor:jucys-eig}
The eigenvalues $\lambda _{\mu }(\alpha )$ of PBT-operator with $N=n-1$ ports
and eigenvalues $\gamma _{\mu }(\alpha )$ of Jucys-Murphy element $J_{n-1}$
are related in a very simple way
\begin{equation}
\lambda _{\mu }(\alpha )-d=(n-1)\frac{m_{\mu }d_{\alpha }}{m_{\alpha }d_{\mu
}}-d=\gamma _{\mu }(\alpha ).
\end{equation}
\end{corollary}

In this way, we get, from PBT formalism, another (up to our best knowledge not known earlier) expression for the spectrum of the JM element $J_{n-1}$. Another expression for the numbers $\eta _{\mu }(\alpha )$, so also
for the eigenvalues $\gamma _{\mu }(\alpha )$ of $J_{n-1}$, is given in~\cite{MozJPA}.

\begin{remark}
Note that although we have very simple relation between PBT operator $\rho$ and JM operator $J_{n}$:
\begin{equation}
J_{n} \mapsto \rho=J_n^{t_n},
\end{equation}
the spectrum of the operator $\rho$ is in very simple 'shift' relation with the spectrum of the 'shorter' JM operator $J_{n-1}$, see equation~\eqref{eq:55}.
\end{remark}

Let us remind basic properties of the regular representation of a group
algebra $\mathbb{C}[G]$ for finite a finite group $G$.

\begin{proposition}
\label{Sec2Prop12}
Let $\psi ^{\mu }$ be any irrep (not necessarily $PRIR$) of a finite group $G$, then the operators 
\begin{equation}
E_{ij}^{\mu }=\frac{d_{\mu }}{|G|}\sum_{g\in G}\psi _{ji}^{\mu }(g^{-1})g\in \mathbb{C}[G]
\end{equation}%
have the following properties
\begin{equation}
E_{ij}^{\mu }E_{kl}^{\nu }=\delta ^{\mu \nu }E_{il}^{\mu },\quad
gE_{ij}^{\mu }=\sum_{k=1}^{d_{\mu }}\psi _{ki}^{\mu }(g)E_{kj}^{\mu },\quad
\forall j=1,\ldots,d_{\mu },\quad 
\mathbb{C}[G]=\bigoplus _{\mu \in \widehat{G}}E^{\mu },
\end{equation}%
where $E^{\mu }=\operatorname{span}_{\mathbb{C}}\{E_{ij}^{\mu }:i,j=1,\ldots,d_{\mu }\}.$ The algebra $\mathbb{C}[G]$ is a direct sum of non-isomorphic matrix algebras and for any
fixed $j=1,\ldots,d_{\mu }$ set of $d_{\mu }$ vectors $E_{ij}^{\mu
}:i=1,\ldots,d_{\mu }$ form a basis of irrep $\psi ^{\mu }$ contained in $\mathbb{C}[G].$ Each subalgebra $E^{\mu }$ contains $d_{\mu }$ such irreps.
\end{proposition}

If the $\psi ^{\mu }\in \mathbb{C}[G]$ are $PRIR^{\prime }s$, then we may rewrite the expressions from Proposition~\ref{Sec2Prop12} as
follows \ 
\begin{equation}
(E_{R}^{\mu })_{i_{\alpha }\quad \quad j_{\beta }}^{\alpha (a_{\alpha
})\beta (b_{\beta })}=\frac{d_{\mu }}{|G|}\sum_{g\in }(\psi _{R}^{\mu
}(g^{-1})_{j_{\beta }\quad \quad i_{\alpha }}^{\beta (b_{\beta })\alpha
(a_{\alpha })}g,\quad 
\end{equation}%
\begin{equation}
g(E_{R}^{\mu })_{j_{\beta }\quad \quad i_{\alpha }}^{\beta (b_{\beta
})\alpha (a_{\alpha })}=\sum_{\gamma (c_{\gamma }),k_{\gamma }}(\psi
_{R}^{\mu }(g)_{k_{\gamma }\quad \quad j_{\beta }}^{\gamma (c_{\gamma
})\beta (b_{\beta })}(E_{R}^{\mu })_{k_{\gamma }\quad \quad i_{\alpha
}}^{\gamma (c_{\gamma })\alpha (a_{\alpha })},
\end{equation}%
\begin{equation}
(E_{R}^{\mu })_{i_{\alpha }\quad \quad j_{\beta }}^{\alpha (a_{\alpha
})\beta (b_{\beta })}(E_{R}^{\mu })_{i_{\alpha }^{\prime }\quad \quad
j_{\beta }^{\prime }}^{\alpha ^{\prime }(a_{\alpha }^{\prime })\beta
^{\prime }(b_{\beta }^{\prime })}=\delta ^{\beta \alpha ^{\prime }}\delta
^{b_{\beta }a_{\alpha }^{\prime }}\delta _{j_{\beta }i_{\alpha }^{\prime
}}(E_{R}^{\mu })_{i_{\alpha }\quad \quad j_{\beta }^{\prime }}^{\alpha
(a_{\alpha })\beta ^{\prime }(b_{\beta }^{\prime })}.
\end{equation}

Now we consider, in the regular representation of the group algebra $\mathbb{C}[G]$, a construction of induced representation $\operatorname{Ind}\uparrow
_{H}^{G}(\varphi ^{\alpha })$, where $\varphi ^{\alpha }\in \widehat{H}$ is
an arbitrary irrep of the subgroup $H\subset G$. We have (for simplicity we
 omit label $\beta $ in indices $i,j$):

\begin{proposition}
Consider the standard matrix algebra generated by the irrep $\varphi
^{\beta }=(\varphi _{ij}^{\beta })\in \widehat{H}$ (here for simplicity we
write omit label $\beta $ in indices $i,j$) 
\begin{equation}
\label{Sec2eq35}
E_{ij}^{\beta }=\frac{d_{\beta }}{|H|}\sum_{h\in H}\varphi _{ji}^{\beta
}(h^{-1})h\in \mathbb{C}[H]\subset \mathbb{C}[G],
\end{equation}
and the transversal $T=\{t_{k}:k=1,\ldots,\frac{|G|}{|H|}\equiv s\}$, then for
any fixed value of $j=1,\ldots,d_{\beta }$, the $sd_{\beta }$ vectors 
\begin{equation}
t_{k}E_{ij}^{\beta }:k=1,\ldots,\frac{|G|}{|H|}\equiv s,\quad i=1,\ldots,d_{\beta }
\end{equation}%
form a basis of the induced representation $\operatorname{Ind}\uparrow _{H}^{G}(\varphi
^{\beta })$ of the group $G$, embedded in the regular representation of group
algebra $\mathbb{C}[G]$, i.e. we have 
\begin{equation}
\forall g\in G\quad g.t_{k}E_{ij}^{\beta }=\sum_{t_{p}\in
T}\sum_{l=1,\ldots,d_{\alpha }}\varphi _{li}^{\beta
}(t_{p}^{-1}gt_{k})t_{p}E_{lj}^{\beta },
\end{equation}%
where the summation over $t_{p}\in T$ is taken over such $t_{p}\in T$, that $%
t_{p}^{-1}gt_{k}\in H$~\cite{Mill}.
From the uniqueness of coset decomposition, eq.~\eqref{Sec2eq13} we get that such an
element $t_{p}\in T$ is unique and we have $
t_{p}^{-1}gt_{k}=h=h_{k}\in H$ so the above equation may written
equivalently
\begin{equation}
\label{eq:induced}
\forall g\in G\quad g.t_{k}E_{ij}^{\beta }=\sum_{l=1,\ldots,d_{\alpha }}\varphi
_{li}^{\beta }(h_{k})t_{p}E_{lj}^{\beta },
\end{equation}
where  $gt_{k}=t_{p}h_{k}$ is the unique coset
decomposition. Thus we see that the main feature of the induced
representation is that the action of $g\in G$ on basis vectors $
\{t_{k}E_{ij}^{\beta }\}$ permutes the transversal  vectors $
\{t_{k}:k=1,\ldots,\frac{|G|}{|H|}\equiv s\}$  and transforms vectors $
\{E_{ij}^{\beta }:i=1,\ldots,d_{\beta }\}$ ($j$ is fixed)
according irrep $\varphi ^{\beta }$.
\end{proposition}

The subspaces 
\begin{equation}
I_{j}^{\beta }=I_{j_{\beta }}^{\beta }=\operatorname{span}_\mathbb{C}\left\{t_{k}E_{ij}^{\beta }:k=1,\ldots,\frac{|G|}{|H|}\equiv s,\quad i=1,\ldots,d_{\beta
}\right\},\quad j_{\beta }=1,\ldots,d_{\beta }
\end{equation}%
are representation spaces for representation $\operatorname{Ind}\uparrow
_{H}^{G}(\varphi ^{\beta })$. From the above we get

\begin{corollary}
We have the following decomposition of algebra $\mathbb{C}[G]$ 
\begin{equation}
\mathbb{C}[G]=\bigoplus _{\beta \in \widehat{H}} \ \bigoplus _{j_{\beta }=1,\ldots,d_{\beta
}}I_{j_{\beta }}^{\beta } \ ,\quad \mathbb{C}[G]=\bigoplus _{\mu \in \widehat{G}}E^{\mu }.
\end{equation}
\end{corollary}

Now it is well known that the induced representation $\operatorname{Ind}\uparrow
_{H}^{G}(\varphi ^{\alpha })$ of the group $G$ is in general reducible. It
appears that the reduction of $\operatorname{Ind}\uparrow _{H}^{G}(\varphi ^{\alpha })$
onto the direct product of irreps of the group $G$ may be
achieved using $PRIR^{\prime }s$.  Using such a $PRIR$ matrices $\psi
_{R}^{\mu }$ we define the following matrices

\begin{definition}
\label{Sec2Def15}
Let $\varphi ^{\beta }\in \widehat{H\text{ }}$ be an irrep and all $\psi
_{R}^{\mu }\equiv \psi _{R}^{\mu _{\beta }}\in \widehat{G}$ be a $
PRIR^{\prime }s$, such that $\varphi ^{\beta }\in \operatorname{Res}\downarrow
_{H}^{G}(\psi _{R}^{\mu })$. For such an irreps $\varphi ^{\beta
}$ and $PRIR^{\prime }s$ $\psi _{R}^{\mu _{\beta }}$ we define a matrix $
U(\beta )$ with coefficients 
\begin{equation}
U(\beta )_{b_{\beta }\quad l_{\alpha },k_{\beta }}^{\mu \alpha (a_{\alpha
}),~t_{k}}=\sqrt{\frac{|H|d_{\mu }}{|G|d_{\beta }}}(\overline{\psi }
_{R}^{\mu })_{l_{\alpha }\quad \quad k_{\beta }}^{\alpha (a_{\alpha })\beta
(b_{\beta })}(t_{k})=U(\beta )_{RL},
\end{equation}
so the left multi-index $R=
\begin{array}{c}
\mu _{\beta }\quad \alpha (a_{\alpha }) \\ 
b_{\beta }\quad l_{\alpha }
\end{array}$, which includes four indices, runs over all $\mu _{\beta }$, over $PRIR$
indices $\alpha (a_{\alpha })$, $l_{\alpha }$ inside $PRIR^{\prime }s$ $\psi
_{R}^{\mu _{\beta }}$, as well over index $b_{\beta }$, which enumerates
the copies of the irrep $\varphi ^{\beta }$ in $\psi _{R}^{\mu _{\beta }}$. The right muli-index $L=
\begin{array}{c}
t_{k} \\ 
k_{\beta }
\end{array}
$ runs over the natural indices \ of the induced representation $\operatorname{Ind}\uparrow
_{H}^{G}(\varphi ^{\beta })$. Matrix $U(\beta )$ is of dimension $d_{\beta
}[G:H]=\dim \operatorname{Ind}\uparrow _{H}^{G}(\varphi ^{\beta })$.  One can check that 
\begin{equation}
(U^{\dagger})(\beta )_{k_{\beta },\quad b_{\beta }~l_{\alpha }}^{t_{k},~~\mu
\alpha (a_{\alpha })}=\sqrt{\frac{|H|d_{\mu }}{|G|d_{\beta }}}(\overline{
\psi }_{R}^{\mu })_{k_{\beta }\quad \quad l_{\alpha }}^{\beta (b_{\beta
})\alpha (a_{\alpha })}(t_{k})=U^{\dagger}(\beta )_{IR}.
\end{equation}
\end{definition}

Matrix $U(\beta )$ is determined by irrep $\varphi ^{\beta }\in 
\widehat{H}$ because, by assumption irrep $\varphi ^{\beta }$
determines all irreps $\mu _{\beta }\in \widehat{G}$,
which must satisfy $\varphi ^{\beta }\in \operatorname{Res}\downarrow _{H}^{G}(\psi
_{R}^{\mu })$. Note that the matrix $U(\beta )$ is not proportional to some $
PRIR^{\prime }s$ matrix $\psi _{R}^{\mu (\beta )}$, the coefficients of the
matrix $U(\beta )$ \ are equal to corresponding coefficients of the
different matrices $\psi _{R}^{\mu (\beta )}$. In general the matrices $
U(\beta )$ and $\psi _{R}^{\mu (\beta )}$ have different dimension.

Directly from Definition~\ref{Sec2Def15} and the sum rule for $PRIRs$ in Proposition~\ref{Sec2Prop2} we
get

\begin{proposition}
The matrix $U(\beta )$ is unitary with respect to the induced multi-indices $%
I=%
\begin{array}{c}
t_{k} \\ 
k_{\beta }%
\end{array}%
$, i.e. we have 
\begin{equation}
\sum_{t_{k},k_{\beta }}U(\beta )_{b_{\beta }\quad l_{\alpha },k_{\beta
}}^{\nu \alpha (a_{\alpha }),~t_{k}}(U^{\dagger})(\beta )_{k_{\beta },~b_{\beta
}^{\prime }~j_{\gamma }}^{t_{k},~~\mu ~\gamma (c_{\gamma })}=\delta ^{\mu
\nu }\delta _{b_{\beta }b_{\beta }^{\prime }}\delta ^{\alpha \gamma }\gamma
^{a_{\alpha }c_{\gamma }}\delta _{l_{\alpha }j_{\gamma }}
\end{equation}%
and it is known that if a matrix is unitary with respect to the columns,
then it is also unitary with respect to the rows, so the matrix $U(\beta
):\varphi ^{\beta }\in \widehat{H}$ is unitary and has dimension $
d_{\beta }\frac{|G|}{|H|}$.
\end{proposition}

Now let us consider subspace $I_{j_{\beta }}^{\beta }\subset \mathbb{C}[G]$ with fixed label $j_{\beta }=1,\label,d_{\beta }$, a representation
space for $\operatorname{Ind}\uparrow _{H}^{G}(\varphi ^{\beta })$, which in general, as
representation of $G$ is reducible. It appears that matrix $U(\beta )$
realize reduction of natural basis $\{t_{k}E_{i_{\beta }j_{\beta }}^{\beta
}\}$ in $I_{j_{\beta }}^{\beta }$ to the reduced one with $PRIR$ operators $
(E^{\mu })_{l_{\alpha }\quad \quad j_\beta}^{\alpha (a_{\alpha })\beta (b_{\beta
})}$ as a basis.

\begin{lemma}
\label{Sec2Le17}
Let $\varphi ^{\beta }\in \widehat{H}$ be an irrep and let us fix label $
j_{\beta }=1,\ldots,d_{\beta }$. Then the unitary matrix $U(\beta )$ transforms
the natural basis $\{t_{k}E_{i_{\beta }j_{\beta }}^{\beta }\}$ of the
subspace $I_{j_{\beta }}^{\beta }\subset \mathbb{C}[G],$  representation space for $\operatorname{Ind}\uparrow _{H}^{G}(\varphi ^{\beta})$ onto a reduced basis of this representation space,  that is we have
\begin{equation}
\label{Sec2eq43}
\sqrt{\frac{|G|}{|H|}\frac{d_{\mu }}{d_{\beta }}}\sum_{t_{p}\in
T}\sum_{k_{\beta }=1,\ldots,d_{\beta }}U(\beta )_{b_{\beta }\quad l_{\alpha
},k_{\beta }}^{\mu \alpha (a_{\alpha }),~t_{p}}t_{p}E_{k_{\beta }j_{\beta
}}^{\beta (b_{\beta })}=(E^{\mu })_{l_{\alpha }\quad \quad j_{\beta
}}^{\alpha (a_{\alpha })\beta (b_{\beta })},
\end{equation}
where 
\begin{equation}
(E^{\mu })_{l_{\alpha }\quad \quad j_{\beta }}^{\alpha (a_{\alpha })\beta
(b_{\beta })}=\frac{d_{\mu }}{|G|}\sum_{g\in G}(\psi _{R}^{\mu })_{j_{\beta
}\quad \quad l_{\alpha }}^{\beta (b_{\beta })\alpha (a_{\alpha
})}(g^{-1})g\in I_{j_{\beta }}^{\beta }
\end{equation}
are standard matrix operators of the group algebra $\mathbb{C}[G]$, defined by the irrep $\psi _{R}^{\mu }$ in the $PRIR
$ version included in $I_{j_{\beta }}^{\beta },$ the representation $
\operatorname{Ind}\uparrow _{H}^{G}(\varphi ^{\beta })$. 
\end{lemma}

\begin{proof}
In order to prove the equation 
\begin{equation}
\frac{|G|}{|H|}\frac{d_{\mu }}{d_{\beta }}\sum_{t_{p}\in T}\sum_{k_{\beta
}=1,\ldots,d_{\beta }}\overline{(\psi _{R}^{\mu })}_{l_{\alpha }\quad \quad
k_{\beta }}^{\alpha (a_{\alpha })\beta (b_{\beta })}(t_{k})t_{p}E_{k_{\beta
}j_{\beta }}^{\beta (b_{\beta })}=(E^{\mu })_{l_{\alpha }\quad \quad
j_{\beta }}^{\alpha (a_{\alpha })\beta (b_{\beta })},
\end{equation}%
it is enough to use the definition of the matrix operators $E_{k_{\beta
}j_\beta}^{\beta }$ and $(E^{\mu })_{l_{\alpha }\quad \quad j_\beta}^{\alpha (a_{\alpha
})\beta (b_{\beta })}$ (see Eq.~\eqref{Sec2eq35}) together with equation~\eqref{Sec2eq17}.
\end{proof}

From  equation~\eqref{Sec2eq43} in  Lemma~\ref{Sec2Le17} we see that each irrep $\varphi
^{\beta (b_{\beta })}$ of $H$, where $\beta (b_{\beta })=1,\ldots,m_{\beta
}^{\mu },$ included in $\operatorname{Res}\downarrow _{H}^{G}(\psi _{R}^{\mu })$
defines one irrep $\psi _{R}^{\mu }\equiv \psi _{R}^{\mu (\beta )}$ (the
index $j_{\beta }$ is fixed) with basis vectors indexed by the left pair of 
$PRIR$ indices $\left( 
\begin{array}{c}
\alpha (a_{\alpha }) \\ 
l_{\alpha }%
\end{array}%
\right) $ in the element $(E^{\mu })_{l_{\alpha }\quad \quad j_{\beta
}}^{\alpha (a_{\alpha })\beta (b_{\beta })}$ and this correspondence is
one-to-one from the invertibility of the matrix $U(\beta )$.
From this as a corollary we get

\begin{theorem}
\label{Sec2Thm18}
Let $\varphi ^{\beta }\in \widehat{H},$  $I_{j_{\beta }}^{\beta }\subset \mathbb{C}[G]$ be the corresponding  representation space for $\operatorname{Ind}\uparrow
_{H}^{G}(\varphi ^{\beta })$ with natural basis $\{t_{p}E_{k_{\beta
}j_{\beta }}^{\beta }\}$  and $U(\beta )$ be an unitary $PRIR$ matrix
determined by $\varphi ^{\beta }$, then
\begin{enumerate}
\item the transformation of the natural basis of the representation $
\operatorname{Ind}\uparrow _{H}^{G}(\varphi ^{\beta })$ in $
\mathbb{C}[G]$ to the reduced one is realised by the $PRIR$ unitary matrix $
U(\beta )$ in the following way
\begin{equation}
\label{eq:ind-red}
U(\beta ):t_{k}E_{i_{\beta }j_{\beta }}^{\beta }\rightarrow \sqrt{\frac{|G|}{%
|H|}\frac{d_{\mu }}{d_{\beta }}}\sum_{t_{p}\in T}\sum_{k_{\beta
}=1,\ldots,d_{\beta }}U(\beta )_{b_{\beta }\quad l_{\alpha },k_{\beta }}^{\mu
\alpha (a_{\alpha }),~t_{p}}t_{p}E_{k_{\beta }j_\beta}^{\beta }=(E^{\mu
})_{l_{\alpha }\quad \quad j_\beta}^{\alpha (a_{\alpha })\beta (b_{\beta })},
\end{equation}

\item from this one can deduce the following decomposition of the induced
representations%
\begin{equation}
\operatorname{Ind}\uparrow _{H}^{G}(\varphi ^{\beta })=\bigoplus _{\mu (\beta )}m_{\beta
}^{\mu }\psi _{R}^{\mu (\beta )},
\end{equation}
so each irrep $\psi _{R}^{\mu }$ appears in representation $\operatorname{Ind}\uparrow
_{H}^{G}(\varphi ^{\beta })$ with multiplicity equal to the multiplicity $
m_{\beta }^{\mu }$ of irrep $\varphi ^{\beta }$ in $\operatorname{Res}\downarrow
_{H}^{G}(\psi _{R}^{\mu }).$
\end{enumerate}
\end{theorem}

The second part of Theorem~\ref{Sec2Thm18} is the famous Frobenius reciprocity theorem~\cite{Curtis,Tullio}. This
classical theorem is usually proved using group character properties and this
proof seems to be the simplest. Our proof of the theorem is technically much
more complicated, but its first result gives in eq. \eqref{eq:ind-red} a unitary matrix $U(\beta )$ which realizes the reduction of the natural basis of induced representation to the
reduced one, and this is the most complicated part of our proof. To our knowledge, this is a new result. 
The standard, character  proof of the Frobenius theorem does not give such
representation reduction (via unitary matrix) because that proof  is
independent of matrix form of considered representations (which is obtained
by use of character theory). The second statement of Theorem \ref{Sec2Thm18}, i.e. the Frobenius
theorem itself is a relatively simple corollary from the multi-index structure
of the matrix $U(\beta ).$

\section{Discussion}
In this paper we discuss the appearance of induced representation in port-based teleportation protocols, focusing on its deterministic version. First, we define the concept of partially reduced irreducible representation for an arbitrary group $G$ and its subgroup $H$ and discuss the most important properties of the introduced concept. This part is presented on the most possible abstract level, i.e. we do not restrict ourselves to a specific choice of the group $G$. Afterward, we choose $G$ to be a symmetric group $S(n)$ with subgroup $H=S(n-1)$ - groups naturally appearing in all variants of the PBT. This restriction allows us to show the relation between basic objects for representation theory for the symmetric group, so-called Jucys-Murphy elements, and the PBT operator $\rho$. In particular, we prove a linear connection between the spectra of these two objects, presenting two different expressions. Spectra of the JM element $J_{n-1}$ and $n-$particle PBT operator are related by a simple shift by factor $d$. In particular, we give new expression for the eigenvalues of the Jucys-Murphy elements based on the irreducible character of the symmetric group, and new orthogonality relations exploiting the concept of PRIRS. In the special case, when one considers the natural representation of Jucys-Murphy elements in dimension 2, we present fully analytical expressions for the mentioned spectrum. Evaluated spectra are crucial in determining entanglement fidelity in deterministic PBT schemes.  At the end, we present a matrix proof of the famous Frobenius reciprocity theorem for characters, and we give the explicit construction of the unitary matrix that realizes the reduction of the natural basis of induced representation to the reduced one.

\section*{Acknowledgements}
This research was funded in whole or in part by the National Science Centre, Poland, Grant number 
     2020/39/D/ST2/01234 (MM, MS).  MH is supported by the National Science Center, Poland
within the QuantERA II Programme (No 2021/03/Y/ST2/00178, acronym ExTRaQT) 
that has received funding from the European Union’s Horizon 2020. For the purpose of Open Access, the author has applied a CC-BY public copyright license to any Author Accepted Manuscript (AAM) version arising from this submission. 
\section*{Conflicts of interest and data availability}
On behalf of all authors, the corresponding author states that there is no conflict of interest and the manuscript has no associated data.

\bibliographystyle{unsrt}
\bibliography{biblio2}

\begin{thebibliography}{10}

\bibitem{ishizaka_asymptotic_2008}
Satoshi Ishizaka and Tohya Hiroshima.
\newblock Asymptotic {Teleportation} {Scheme} as a {Universal} {Programmable} {Quantum} {Processor}.
\newblock {\em Physical Review Letters}, 101(24):240501, December 2008.

\bibitem{ishizaka_quantum_2009}
Satoshi Ishizaka and Tohya Hiroshima.
\newblock Quantum teleportation scheme by selecting one of multiple output ports.
\newblock {\em Physical Review A}, 79(4):042306, April 2009.

\bibitem{beigi2011simplified}
Salman Beigi and Robert K{\"o}nig.
\newblock Simplified instantaneous non-local quantum computation with applications to position-based cryptography.
\newblock {\em New Journal of Physics}, 13(9):093036, 2011.

\bibitem{may2022complexity}
Alex May.
\newblock Complexity and entanglement in non-local computation and holography.
\newblock {\em Quantum}, 6:864, 2022.

\bibitem{buhrman_quantum_2016}
Harry Buhrman, {\L}ukasz Czekaj, Andrzej Grudka, Micha{\l} Horodecki, Pawe{\l} Horodecki, Marcin Markiewicz, Florian Speelman, and Sergii Strelchuk.
\newblock Quantum communication complexity advantage implies violation of a {Bell} inequality.
\newblock {\em Proceedings of the National Academy of Sciences}, 113(12):3191--3196, March 2016.

\bibitem{pirandola2019fundamental}
Stefano Pirandola, Riccardo Laurenza, Cosmo Lupo, and Jason~L Pereira.
\newblock Fundamental limits to quantum channel discrimination.
\newblock {\em npj Quantum Information}, 5(1):50, 2019.

\bibitem{pereira2021characterising}
Jason Pereira, Leonardo Banchi, and Stefano Pirandola.
\newblock Characterising port-based teleportation as universal simulator of qubit channels.
\newblock {\em Journal of Physics A: Mathematical and Theoretical}, 54(20):205301, 2021.

\bibitem{quintino2021quantum}
Marco~T{\'u}lio Quintino.
\newblock Quantum teleportation beyond its standard form: Multi-port-based teleportation.
\newblock {\em Quantum Views}, 5:56, 2021.

\bibitem{PhysRevLett.122.170502}
Michal Sedl\'ak, Alessandro Bisio, and M\'ario Ziman.
\newblock Optimal probabilistic storage and retrieval of unitary channels.
\newblock {\em Phys. Rev. Lett.}, 122:170502, May 2019.

\bibitem{wang_higher-dimensional_2016}
Zhi-Wei Wang and Samuel~L. Braunstein.
\newblock Higher-dimensional performance of port-based teleportation.
\newblock {\em Scientific Reports}, 6:33004, September 2016.

\bibitem{strelchuk_generalized_2013}
Sergii Strelchuk, Micha{\l} Horodecki, and Jonathan Oppenheim.
\newblock Generalized {Teleportation} and {Entanglement} {Recycling}.
\newblock {\em Physical Review Letters}, 110(1):010505, January 2013.

\bibitem{Studzinski2017}
Micha{\l} {Studzi{\'n}ski}, Sergii {Strelchuk}, Marek {Mozrzymas}, and Micha{\l} {Horodecki}.
\newblock {Port-based teleportation in arbitrary dimension}.
\newblock {\em Scientific Reports}, 7:10871, Sep 2017.

\bibitem{StuNJP}
Marek {Mozrzymas}, Micha{\l} {Studzi{\'n}ski}, Sergii {Strelchuk}, and Micha{\l} {Horodecki}.
\newblock {Optimal port-based teleportation}.
\newblock {\em New Journal of Physics}, 20(5):053006, May 2018.

\bibitem{MozJPA}
Marek {Mozrzymas}, Micha{\l} {Studzi{\'n}ski}, and Micha{\l} {Horodecki}.
\newblock {A simplified formalism of the algebra of partially transposed permutation operators with applications}.
\newblock {\em Journal of Physics A Mathematical General}, 51(12):125202, Mar 2018.

\bibitem{christandl2021asymptotic}
Matthias Christandl, Felix Leditzky, Christian Majenz, Graeme Smith, Florian Speelman, and Michael Walter.
\newblock Asymptotic performance of port-based teleportation.
\newblock {\em Communications in Mathematical Physics}, 381:379--451, 2021.

\bibitem{Leditzky2022}
Felix Leditzky.
\newblock Optimality of the pretty good measurement for port-based teleportation.
\newblock {\em Letters in Mathematical Physics}, 112(5):98, Sep 2022.

\bibitem{grinko2023gelfandtsetlin}
Dmitry Grinko, Adam Burchardt, and Maris Ozols.
\newblock Gelfand-tsetlin basis for partially transposed permutations, with applications to quantum information, arXiv: 2310.02252 (2023).

\bibitem{StuIEEE}
Michał Studziński, Marek Mozrzymas, Piotr Kopszak, and Michał Horodecki.
\newblock Efficient multi port-based teleportation schemes.
\newblock {\em IEEE Transactions on Information Theory}, 68(12):7892--7912, 2022.

\bibitem{Koch_2012}
R~de~Mello~Koch, N~Ives, and M~Stephanou.
\newblock On subgroup adapted bases for representations of the symmetric group.
\newblock {\em Journal of Physics A: Mathematical and Theoretical}, 45(13):135204, mar 2012.

\bibitem{Mu}
G.E Murphy.
\newblock A new construction of young's seminormal representation of the symmetric groups.
\newblock {\em Journal of Algebra}, 69(2):287--297, 1981.

\bibitem{Ju}
A.-A.A. Jucys.
\newblock Symmetric polynomials and the center of the symmetric group ring.
\newblock {\em Reports on Mathematical Physics}, 5(1):107--112, 1974.

\bibitem{VGTuraev_1990}
V~G Turaev.
\newblock Operator invariants of tangles, and r-matrices.
\newblock {\em Mathematics of the USSR-Izvestiya}, 35(2):411, apr 1990.

\bibitem{KOIKE198957}
Kazuhiko Koike.
\newblock On the decomposition of tensor products of the representations of the classical groups: By means of the universal characters.
\newblock {\em Advances in Mathematics}, 74(1):57--86, 1989.

\bibitem{BENKART1994529}
G.~Benkart, M.~Chakrabarti, T.~Halverson, R.~Leduc, C.Y. Lee, and J.~Stroomer.
\newblock Tensor product representations of general linear groups and their connections with brauer algebras.
\newblock {\em Journal of Algebra}, 166(3):529--567, 1994.

\bibitem{BEN96}
Georgia Benkart.
\newblock Commuting actions—a tale of two groups.
\newblock {\em In: Lie Algebras and Their Representations. Contemporary Mathematics}, 194, 1996.

\bibitem{bulgakova:tel-02554375}
Daria~V. Bulgakova.
\newblock {\em {Some Aspects Of Representation Theory Of Walled Brauer Algebras}}.
\newblock Theses, {Aix Marseille Universit{\'e}}, January 2020.

\bibitem{Bra37}
Richard Brauer.
\newblock On algebras which are connected with the semisimple continuous groups.
\newblock {\em Annals of Mathematics}, 38(4):857--872, 1937.

\bibitem{stud2020A}
Michał Studziński, Marek Mozrzymas, Piotr Kopszak, and Michał Horodecki.
\newblock Efficient multi-port teleportation schemes.
\newblock 2020.

\bibitem{Studzinski_2022}
Michał Studziński, Marek Mozrzymas, and Piotr Kopszak.
\newblock Square-root measurements and degradation of the resource state in port-based teleportation scheme.
\newblock {\em Journal of Physics A: Mathematical and Theoretical}, 55(37):375302, aug 2022.

\bibitem{mozrzymas2021optimal}
Marek Mozrzymas, Micha{\l} Studzi{\'n}ski, and Piotr Kopszak.
\newblock Optimal multi-port-based teleportation schemes.
\newblock {\em Quantum}, 5:477, 2021.

\bibitem{Ram1992MatrixUF}
Arun Ram and Hans Wenzl.
\newblock Matrix units for centralizer algebras.
\newblock {\em Journal of Algebra}, 145:378--395, 1992.

\bibitem{vershikBook}
F.Scarabotti T.~Ceccherini-Silberstein and F.~Tolli.
\newblock {\em Representation Theory of the Symmetric Group. The Okounkov-Vershik Approach, Character Formulas, and Partition Algebras}.
\newblock Cambridge University Press, New York, 2010.

\bibitem{Okunkov}
Andrei Okounkov and Anatoly Vershik.
\newblock A new approach to representation theory of symmetric groups.
\newblock {\em Selecta Mathematica}, 2(4):581, Sep 1996.

\bibitem{Vershik2005}
A.~M. Vershik and A.~Yu. Okounkov.
\newblock A new approach to the representation theory of the symmetric groups. ii.
\newblock {\em Journal of Mathematical Sciences}, 131(2):5471--5494, Nov 2005.

\bibitem{Fulton1991-book-rep}
W.~Fulton and J.~Harris.
\newblock {\em Representation Theory - A first Course}.
\newblock Springer-Verlag, New York, 1991.

\bibitem{Tullio}
Tullio Ceccherini-Silberstein, Fabio Scarabotti, and Filippo Tolli.
\newblock {\em Representation {Theory} and {Harmonic} {Analysis} of {Wreath} {Products} of {Finite} {Groups}}.
\newblock Cambridge University Press, January 2014.

\bibitem{Mill}
Willard~Miller Jr.
\newblock {\em Symmetry groups and their applications}.
\newblock Academic Press New York and London, January 1972.

\bibitem{Curtis}
Charles~W. Curtis and Irving Reiner.
\newblock {\em Representation Theory of Finite Groups and Associative Algebras}.
\newblock John Wiley and Sons, New York, 1988.

\end{thebibliography}
\end{document}